\documentclass{lmcs} 
\pdfoutput=1

\usepackage{lastpage}
\lmcsdoi{17}{1}{16}
\lmcsheading{}{\pageref{LastPage}}{}{}%
{Dec.~16,~2019}{Feb.~18,~2021}{}


\usepackage{graphicx}
\usepackage[utf8]{inputenc}
\usepackage{amsbsy,amssymb,amsmath}
\usepackage{color}
\usepackage{booktabs}
\usepackage{listings}

\usepackage[english]{babel}

\definecolor{dkblue}{rgb}{0,0.1,0.5}
\definecolor{lightblue}{rgb}{0,0.5,0.5}
\definecolor{dkgreen}{rgb}{0,0.4,0}
\definecolor{dk2green}{rgb}{0.4,0,0}
\definecolor{dkviolet}{rgb}{0.6,0,0.8}

\newcommand{\Coq}{{\sc Coq}}
\newcommand{\MC}{{\sc Mathematical Components}}
\newcommand{\Ocaml}{{\sc Ocaml}}
\newcommand{\CoqEAL}{{\sc CoqEAL}}
\newcommand{\HOLLight}{{\sc HOL-Light}}
\newcommand{\IsaHOL}{{\sc Isabelle/HOL}}
\newcommand{\Mgfun}{{\sc Mgfun}}
\newcommand{\Algolib}{{\sc Algolib}}

\newcommand{\C}[1]{\mbox{\lstinline`#1`}}

\newcommand{\bN}{\mathbb{N}}
\newcommand{\bR}{\mathbb{R}}
\newcommand{\bQ}{\mathbb{Q}}
\newcommand{\bZ}{\mathbb{Z}}

\newcommand{\bQbar}{{\overline{\mathbb{Q}}}}

\newcommand{\cA}{\mathcal{A}}

\newcommand{\bigO}{\mathcal{O}}

\newcommand*\floor[1]{\left\lfloor{#1}\right\rfloor}

\definecolor{dkblue}{rgb}{0,0.1,0.5}
\definecolor{lightblue}{rgb}{0,0.5,0.5}
\definecolor{dkgreen}{rgb}{0,0.4,0}
\definecolor{dk2green}{rgb}{0.4,0,0}
\definecolor{dkviolet}{rgb}{0.6,0,0.8}

\theoremstyle{plain} 

\begin{document}

\lstset{moredelim=[is][\color{red}\bfseries\ttfamily\underbar]{|*}{*|}}

\title{A Formal Proof of the Irrationality of $\zeta(3)$}

\author[A.~Mahboubi]{Assia Mahboubi}	
\author[T.~Sibut-Pinote]{Thomas Sibut-Pinote}	

\address{LS2N
UFR Sciences et Techniques, 2 rue de la Houssinière, BP 92208 44322
Nantes Cedex 3 France}	
\email{Assia.Mahboubi@inria.fr}  
\email{thomas.sibut-pinote@ens-lyon.org}  

\thanks{This work was supported in part by the project FastRelax
  ANR-14-CE25-0018-01.}	





\begin{abstract}

This paper presents a complete formal verification of a proof that
the evaluation of the Riemann zeta function at 3 is irrational, using
the  \Coq{} proof assistant. This result was first presented by Apéry in
1978, and the proof we have formalized essentially follows the path of his
original presentation. The crux of this proof is to establish that
some sequences satisfy a common recurrence. We formally prove this
result by an \emph{a posteriori} verification of calculations performed by
computer algebra algorithms in a Maple session. The rest of the proof
combines arithmetical ingredients and asymptotic analysis, which we
conduct by extending the \MC{} libraries. 
\end{abstract}

\maketitle


\section{Introduction}
\label{sec:intro}
In 1978, Apéry proved that $\zeta(3)$, which is the sum $\sum_{i
  =1}^{\infty} \frac{1}{i^3}$ now known as the \emph{Apéry constant},
is irrational. This result was the first dent in the problem of the
irrationality of the evaluation of the Riemann zeta function at
\emph{odd} positive integers. As of today, this problem remains a
long-standing challenge of number theory.
Zudilin~\cite{Zudilin-2001-ONZ} showed that at least one of the
numbers $\zeta(5)$,$\zeta(7)$,$\zeta(9),\zeta(11)$ must be irrational.
Ball and Rivoal~\cite{MR1787183,MR1859021} established that there are
infinitely many irrational odd zeta values. Fischler, Sprang and
Zudilin proved~\cite{fischler_sprang_zudilin_2019} that there are
asymptotically more than any power of $\log(s)$ irrational values of
the Riemann zeta function at odd integers between $3$ and $s$. But
today $\zeta(3)$ is the only \emph{known} such value to be irrational.

Van der
Poorten reports~\cite{vanderPoorten-1979-PEM} that Apéry's
announcement of this result was at first
met with wide skepticism.  His obscure presentation featured ``a sequence of
unlikely assertions'' without proofs, not the least of which was an
enigmatic recurrence (Lemma~\ref{lem:rec-apery}) satisfied by two
sequences $a$ and $b$. It took two months of collaboration between
Cohen, Lenstra, and Van der Poorten, with the help of Zagier, to
obtain a thorough proof of Apéry's theorem:

\begin{thm}[Apéry, 1978]\label{thm:apery}
        The constant $\zeta(3)$ is irrational.
\end{thm}

In the present paper, we describe a formal proof of this theorem
inside the \Coq{} proof
assistant~\cite{the_coq_development_team_2020_3744225}, using
the \MC{} libraries~\cite{coqfinitgroup}. This formalization follows
the structure of Apéry's original proof. However, we replace the
manual verification of recurrence relations by an automatic discovery
of these equations, using symbolic computation. For this purpose, we
use Maple packages to perform calculations outside the proof
assistant, and we verify a posteriori the resulting claims inside
\Coq{}. By combining these verified results with additional formal
developments, we obtain a complete formal proof of
Theorem~\ref{thm:apery}, formalized using the \Coq{} proof assistant
without additional axiom. In particular, the proof is entirely
constructive, and does \emph{not} rely on the axiomatic definition of
real numbers provided in \Coq{}'s standard library. A previous
paper~\cite{chyzak:hal-00984057} reported on the implementation of the
cooperation between a computer algebra system and a proof assistant
used in the formalization. The present paper is self-contained: it
includes a summary of the latter report, and provides more details
about the rest of the formal proof. In particular, it describes the
formalization of an upper bound on the asymptotic behavior of $lcm(1,
... , n)$, the least common multiple of the integers from $1$ to $n$,
a part of the proof which was missing in the previous report.

The rest of the paper is organized as follows.
We first describe the background formal theories used in our
development (Section~\ref{sec:preliminaries}).
We then outline the proof of Theorem~\ref{thm:apery} (Section~\ref{sec:proof}).
We summarize the algorithms used in the Maple session, the data this
session produces and the way this data can be used in formal proofs
(Section~\ref{sec:algos}).
We then describe the proof of the consequences of Apéry's recurrence (Section
~\ref{sec:asymptotics}).
Finally, we present an elementary proof of the bound on the asymptotic
behavior of the sequence $lcm(1, ... , n)$, which is used in this
irrationality proof (Section~\ref{sec:hanson}), before commenting on
related work and concluding
(Section~\ref{sec:conclusions_formal_proof}).

The companion code to the present article can be found in the following repository:

\centerline{\url{https://github.com/math-comp/apery}.}

\section{Preliminaries}
\label{sec:preliminaries}

This section provides some hints about the representation of the
different natures of numbers at stake in this proof in the libraries
backing our formal development. It also describes a few extensions we
devised for these libraries and sets some notations used throughout
this paper. Most of the material presented here is related to the
\MC{} libraries~\cite{coqfinitgroup,gonthier:hal-00816699}.

\subsection{Integers}
\label{subsec:nat}

In \Coq{}, the set  $\bN$ of natural numbers is usually represented
by the type~\lstinline+nat+:
\begin{lstlisting}
Inductive |*nat*| := O | S : nat -> nat.
\end{lstlisting}
This type is defined in a prelude library, which is automatically
imported by any \Coq{} session. It models the elements of $\bN$ using
a unary representation: \Coq{}'s parser reads the number $2$ as the term
\lstinline|S (S O)|. The structural induction principle associated with this
inductive type coincides with the usual recurrence scheme on natural
numbers. This is convenient for defining elementary functions on
natural numbers, like comparison or arithmetical operations, and for
developing their associated theory. However, the resulting programs
are usually very naive and inefficient implementations, which should
only be evaluated for the purpose of small scale computations.

The set $\bZ$ of integers can be represented by gluing together two
copies of type~\C{nat}, which provides a signed unary representation
of integers:
\begin{lstlisting}
Inductive |*int*| : Set := Posz of nat | Negz of nat.
\end{lstlisting}
If the term \C{n : nat} represents the natural number $n\in\bN$, then
the term \C{(Posz n) : int} represents the integer $n\in\bZ$ and
the term \C{(Negz n) : int} represents the integer $-(n + 1)\in\bZ$.
In particular, the constructor \lstinline|Posz : nat -> int| implements the
embedding of type \C{nat} into type \C{int}, which is invisible on paper
because it is just the inclusion $\bN\subset\bZ$. In order to mimic
the mathematical practice, the constant \C{Posz} is declared as a
\emph{coercion}, which means in particular that unless otherwise
specified, this function is hidden from the terms displayed by \Coq{}
to the user (in the current goal, in answers to search queries, etc).

The \MC{} libraries provide formal definitions of a few elementary
concepts and results from number theory, defined on the type
\C{nat}. For instance, they provide the theory of Euclidean division,
a boolean primality test, the elementary properties of the factorial
function, of binomial coefficients, etc. In the rest of the paper,
we use the standard mathematical notations $n!$ and $\binom{n}{m}$ for
the corresponding formal definition of the factorial and of the binomial
coefficients respectively. These libraries also define the $p$-adic valuation $v_p(n)$
of a number $n$: if $p$ is a prime number, it is the exponent of $p$
in the prime decomposition of $n$. However, we had to extend the
available basic formal theory with a few extra standard results,
like the formula giving the $p$-adic valuation of factorials:
\begin{lem}\label{lem:valfact}
For any $n \in \bN$ and for any prime number $p$:
$$v_p(n!) = \sum_{i =1}^{\floor{\log_p n}}\floor{\frac{n}{p^i}}.$$
\end{lem}
Incidentally, the formal version of this formula is a typical example
of the slight variations one may introduce in a mathematical
statement, in order to come up with a formal sentence which is not
only correct and faithful to the original mathematical result, but
also a tool which is easy to use in subsequent formal proofs.  First,
although the fraction in the original statement of
Lemma~\ref{lem:valfact} may suggest that rational numbers play a role
here, $\floor{\frac{n}{m}}$ is in fact exactly the quotient of the
Euclidean division of $n$ by $m$. In the rest of the paper, for
$n,m\in \bN$ and $m$ non-zero, we thus write $\floor{\frac{n}{m}}$ for
the quotient of the Euclidean division of $n$ by $m$. Perhaps more
interestingly, the formal statement of Lemma~\ref{lem:valfact} rather
corresponds to the following variant:
$$\textrm{For any prime }p \textrm{ and any }j, n\in\bN, \textrm{ such
  that }n < p ^ {j+1}, v_p(n!) = \sum_{i =1}^j \floor{\frac{n}{p^i}}.$$
Adding an extra variable to generalize the upper bound of the sum is a
better option because it will ease unification when this formula is
applied or used for rewriting. Moreover, we do not really need to
introduce logarithms here: indeed, $\floor{\log_p n}$ is used to
denote the largest power of $p$ smaller than $n$. For this purpose, we
could use the function \C{trunc_log : nat -> nat -> nat} provided by
the \MC{} libraries, which computes the greatest exponent $\alpha$
such that $n^\alpha \leq m$, in other words $\floor{\log_n m}$. Better
yet, since the summand is
zero when the index $i$ exceeds this value, we can simplify the side
condition on the extra variable and require only that $n < p ^ {j+1}$.

The basic theory of binomial coefficients present in the \MC{}
libraries describes their role in elementary enumerative
combinatorics. However, when viewing binomial coefficients as a
sequence which is a certain solution of a recurrence system, it
becomes natural to extend their domain of definition to integers: we
thus developed a small library about these generalized binomial
coefficients. We also needed to extend these libraries with properties
of multinomial coefficients.  For $n, k_1, \dotsc, k_l\in \bN$, with
$k_1 + \dotsb + k_l = n$, the coefficient of $x_1^{k_1} \dotsm
x_l^{k_l}$ in the formal expansion of $\left(x_1 + \dotsb +
x_l\right)^n$ is called a \emph{multinomial coefficient} and denoted
$\binom{n}{k_1, \dotsc, k_l}$. Its value is $\frac{n!}{k_1! \dotsc
  k_l!}$ or equivalently, $\prod\limits_{i=1}^l \binom{k_1 + \dots +
  k_i}{k_i}$. The latter formula provides for free the fact that
multinomial coefficients are non-negative integers and we use it
in our formal definition: for \C{(l : seq nat)} a finite sequence
$l_1,\dots,l_s$ of natural numbers, then \C{(multinomial l)} is the
multinomial coefficient $\binom{l_1+\dots+l_s}{l_1, \dotsc, l_s}$:
\begin{lstlisting}
Definition |*multinomial*| (l : seq nat) : nat :=
 \prod_(0 <= i < size l) (binomial (\sum_(0 <= j < i.+1) l`_j) l`_i).
\end{lstlisting}
From this definition, we prove formally the other
characterizations, as well as the generalized Newton formula
describing the expansion of $\left(x_1 + \dotsb + x_l\right)^n$.

\subsection{Rational numbers, algebraic numbers, real numbers}
\label{ssec:realcomplex}

In the \MC{} libraries, rational numbers are represented using a
dependent pair. This type construct, also called $\Sigma$-type, is
specific to dependent type theory: it makes possible to define a type
that decorates a data with a proof that a certain property holds on
this data. The \MC{} libraries also include a construction of the
algebraic closure for countable fields, and thus a construction of
$\bQbar$, algebraic closure of $\bQ$, the field of rational
numbers. The corresponding type, named \C{algC} is equipped with a
structure of (partially) ordered, algebraically closed
field~\cite{gonthier:hal-00816699}. Slightly abusing notation, we
denote by $\bQbar \cap \bR$ the subset of $\bQbar$ containing elements
with a zero imaginary part and we call such elements real algebraic
numbers.

Almost all the irrational numbers involved in the present proof are
real algebraic numbers, and more precisely, they are of the form
$r^{\frac{1}{n}}$ for $r$ a non-negative rational number $r$ and for $n \in
\bN$. The only place where these numbers play a role is in auxiliary
lemmas for the proof of the asymptotic behavior of the sequence
$(\ell_n)_{n\in\bN}$, where $\ell_n$ is the least common multiple of
integers between $1$ and $n$. They appear in inequalities expressing
signs and estimations.

It might come as a surprise that we used the type
\C{algC} of algebraic (complex) numbers to cast these quantities,
although we do not actually need imaginary complex numbers. But this choice
proved convenient due to the fact that the type \C{algC} features both
a definition of $n$-th roots, and a clever choice of partial
order. Indeed, although $\bQbar$ cannot be ordered as a field, it is
equipped with a binary relation, denoted $\leq$, which coincides with
the real order relation on $\bQbar \cap \bR$:
\[\forall x,y \in \bQbar, \; x \leq y \Leftrightarrow y - x \in \bR_{\geq0}. \]
In particular, for any $z\in \bQbar$:
\[ 0 \leq z \Leftrightarrow z \in {\bR_{\geq 0}} \textrm{ and } z \leq 0 \Leftrightarrow z \in {\bR_{\leq 0}}. \]
Moreover, the type \C{algC} is equipped with a function \C{n.-root :
  algC -> algC}, defined for any \C{(n : nat)}, such that \C{(n.-root
  z)} is the $n$-th (complex) root of \C{z} with minimal non-negative
argument. Crucially, when \C{(z : algC)} represents a non-negative
real number, \C{(n.-root z)} coincides with the definition of the real
$n$-th root, and thus the following equivalence holds:
\begin{lstlisting}
Lemma |*rootC_ge0*| (n : nat) (z : algC) : n > 0 -> (0 <= n.-root z) = (0 <= z).
\end{lstlisting}
The shape of Lemma \C{|*rootC_ge0*|} is typical of the style pervasive
in the \MC{} libraries, where equivalences between decidable
statements are stated as boolean equalities. It expresses that
for an algebraic number $x$, that is for $x \in \bQbar$, we have
$x^\frac{1}{n} \in \bR_{\geq 0}$ if and only if $x \in \bR_{\geq 0}$.

The one notable place at which we need to resort to a larger subset of
the real numbers is the definition of the number $\zeta(3)$, if only
because as of today, it is not even known whether $\zeta(3)$ is
algebraic or transcendental. This number is actually defined as the
limit of the partial sums $\sum_{m=1}^n \frac1{m^3}$, so we start our
formal study by establishing the existence of this limit.

More precisely, we show that this sequence of partial sums is a
\emph{Cauchy sequence}, with an explicit modulus of convergence.

\begin{defi}\label{def:cauchyseq}
  A \emph{Cauchy sequence} is a  sequence
  $(x_n)_{n \in \mathbb{N}} \in \mathbb{Q}^{\mathbb{N}}$ together with a modulus
  of convergence $m_x\ :\ \bQ \rightarrow \bN$ such that if $\varepsilon$ is a positive rational
number, then any two elements of index greater than $m_x(\varepsilon)$
are separated at most by $\varepsilon$.
\end{defi}

\begin{prop}\label{prop:z3def}
The sequence $z_n = \sum_{m=1}^n \frac1{m^3}$ is a Cauchy sequence.
\end{prop}

Two sequences $x$ and $y$ are \emph{Cauchy equivalent} if $x$ and $y$
are both Cauchy sequences, and if eventually
$|x_n - y_n| < \varepsilon$, for any $\varepsilon > 0$. Real numbers
could be constructed formally by introducing a quotient type, whose
element are the equivalent classes of the latter relation. But this is
rather irrelevant for this formalization, which involves explicit
sequences and their asymptotic properties, rather than real
numbers. For this reason, the formal statements in this formal proof
only involve sequences of rational numbers, and a type of Cauchy
sequences which pairs such a sequence with a proof that it has the
Cauchy property.  For instance, for two Cauchy sequences $x$ and $y$,
we write $x < y$ if there is a rational number $\varepsilon > 0$ such
that eventually $x_n + \varepsilon \leq y_n$.

We benefited from the formal library about Cauchy sequences developed
by Cohen~\cite{cohen:hal-00671809}. This library defines Cauchy
sequences of elements in a totally ordered field, and introduces a
type \C{(creal F)} of Cauchy sequences over the totally ordered field
\C{F}, given as a parameter. We thus use the instance \C{(creal rat)}.
The infix notation \C{==}, in the notation scope \C{CR}, denotes the
equivalence of Cauchy sequences, as in the statement \C{(x == y)\%CR},
which states that the two Cauchy sequences \C{x, y} are equivalent. The
library implements a \emph{setoid} of field operations over this
type~\cite{DBLP:journals/jfp/BartheCP03,DBLP:journals/jfrea/Sozeau09},
so as to facilitate substitutions for equivalents in formulas. In
addition, the library provides a tactic called \C{bigenough}, which
eases formal proofs by allowing a dose of laziness. This tactic is
specially useful in proofs that a certain property on sequences is
eventually true, which involve constructing effective moduli of
convergence.

The formal statement corresponding to Theorem~\ref{thm:apery} is thus:
\begin{lstlisting}
Theorem |*zeta_3_irrational*| : ~ exists (r : rat), (z3 == r%:CR)%CR.
\end{lstlisting}
where the postfix notation \C{r\%:CR} denotes the Cauchy sequence
whose elements are all equal to the rational number
\C{(r : rat)}. The term \C{z3} is the Cauchy sequence corresponding to
the partial sums  $(z_n)_{n\in\bN}$, that is, the dependent pair of
this sequence with a proof of Property~\ref{prop:z3def}.
The formal statement thus expresses that no constant rational
sequence can be Cauchy equivalent to $(z_n)_{n\in\bN}$. Interestingly,
a long-lasting typo has marred the formal statement of theorem 
\C{zeta_3_irrational} in the corresponding \Coq{} libraries,
until writing the revised version of the present paper. Until then,
the (inaccurate) statement was indeed:
\begin{lstlisting}
Theorem |*incorrect_zeta_3_irrational*| : ~ exists (r : rat), (z3 = r%:CR)%CR.
\end{lstlisting}
Replacing \C{==} by \C{=} changes the statement completely, as it now
expresses that there is no constant sequence of rationals equal to
\C{z3}: and this is trivially true. The typo was already present in
the version of the code that we made public for our previous
report~\cite{chyzak:hal-00984057}, and the typo has remained
unnoticed since. Yet fortunately, the \emph{proof script} was right, and
actually described a correct proof of the stronger statement
\C{zeta_3_irrational}.

\subsection{Notations}
\label{subsec:theories_notations}
In this section, we provide a few hints on the notations used in the
formal statements corresponding to the paper proof, so as to make
precise the meaning of the statements we have proved formally.
Indeed, this development makes heavy use of the notation facilities
offered by the \Coq{} proof assistant, so as to improve the
readability of formulas. For instance, notation scopes allow to use
the same infix notation for a relation on type \C{nat}, and in this
case \C{(x < y)} unfolds to \C{(ltn x y)\ : bool}, or for a relation
on type \C{creal rat)}, and in that case \C{(x < y)} unfolds to
\C{(lt\_creal x y)\ : Prop}, the comparison predicate described in
Section~\ref{ssec:realcomplex}. Notation scopes can be made explicit
using post-fixed tags: \C{(x < y)\%N} is interpreted in the scope
associated with natural numbers, and \C{(x < y)\%CR}, in the scope
associated with Cauchy sequences.
 
Generic notations can also be shared thanks to type-class like
mechanisms. The \MC{} libraries feature a hierarchy of algebraic
structures~\cite{garillot:inria-00368403}, which organizes a corpus of
theories and notations shared by all the instances of a same
structure.  This hierarchy implements inheritance and sharing using
\Coq{}'s mechanisms of coercions and of canonical
structures~\cite{mahboubi:hal-00816703}. Each structure in the
hierarchy is modeled by a dependent record, which packages a type with
some operations on this type and with requirements on these
operations. In order to equip a given type with a certain structure,
one has to endow this type with enough operations and properties,
following the signature prescribed by the structure.  For example,
these structures are all \emph{discrete}, which means that they
require a boolean equality test. In turn, all instances of all these
structures share the same infix notation \C{(x == y)} for the latter
boolean equality test between \C{x} and \C{y}: this notation makes
sense for \C{x,y} in type \C{nat},\C{int}, \C{rat}, \C{alC},
etc. because all these types are instances of the same structure.
For instance, although type \C{rat} is a dependent pair (see
Section~\ref{ssec:realcomplex}), the boolean comparison test only needs
to work with the data: by Hedberg's theorem~\cite{Hedberg:1998}, the
proof stored in the proof component can be made irrelevant. Note
that the situation is different for the type \C{(creal rat)} of Cauchy
sequences. The formal statement of theorem \C{zeta_3_irrational} (see
Section~\ref{ssec:realcomplex}) uses the same \C{==} infix symbol,
but in a different scope, in which it refers to \C{Prop}-valued Cauchy
equivalence. Indeed, this relation cannot be turned constructively
into a boolean predicate, as the comparison of Cauchy sequences is
not effective. The type \C{algC} of algebraic numbers by contrast
enjoys the generic version of the notation, as ordered fields only
require a partial order relation.

Partial order, but also units of a ring, and inverse operations are
examples of operations involved in some structures of the hierarchy,
that make sense only on a subset of the elements of the carrier. In
the dependent type theory implemented by \Coq{}, it would be possible
to use a dependent pair in order to model the source type of such an
inverse operation. Instead, as a rule of thumb, the signature of a
given structure avoids using rich types as the source types of their
operations but rather ``curry'' the specification. For instance, the
source type of the inverse operation in the structure of ring with
units is the carrier type itself, but the signature of this structure
also has a boolean predicate, which selects the units in this carrier
type. The inverse operation has a default behavior outside units and
the equations of the theory that involve inverses are typically
guarded with invertibility conditions. Hence although the expression
\C{x^-1 * x} is syntactically well-formed for any term \C{x} of an
instance of ring with units, it can be rewritten to \C{1} only when
\C{x} is known to be invertible.

The readability of formulas also requires dealing in a satisfactory
manner with the inclusion of the various collections of numbers at
stake, that are represented with distinct types, for instance:
$$\bN \subset \bZ \subset \bQ \subset \bQbar.$$

The implementation of the inclusion $\bN \subset \bZ$ was mentioned
in Section~\ref{subsec:nat}. The canonical embedding of type \C{int}
is available in the generic theory of rings, but unfortunately, it
cannot be declared as a coercion, and eluded in formal statements: the
type of the corresponding constant would violate the uniform
inheritance condition prescribed by \Coq{}'s coercion
mechanism~\cite{the_coq_development_team_2020_3744225}. Its formal
definition hence comes with a generic postfix notation \lstinline|_%:~R|,
modeled as a reminiscence of the syntax of notation scopes and used
to cast an integer as an element of another ring. The latter embedding
is pervasive in formulas expressing the recurrence relations involved
in this proof. Indeed, these recurrence relations feature polynomial
coefficients in their indices and relate the rational elements of
their solutions. See for instance Equation~\ref{eq:apery-rec}.

\subsection{Computations}
\label{subsec:computations}

Using the unary representation of integers described in
Section~\ref{subsec:nat}, the command:
\begin{lstlisting}
Compute 100*1000.
\end{lstlisting}
which asks \Coq{} to evaluate this product, triggers a stack overflow.
For the purpose of running computations inside \Coq{}'s logic, on
integers of a medium size, an alternate data-structure is required,
together with less naive implementations of the arithmetical
operations. The present formal proof requires this nature of
computations at several places, for instance in order to evaluate
sequences defined by a recurrence relation at a few particular values.
For these computations, we used the binary representation of integers
provided by the \C{ZArith} library included in the standard
distribution of \Coq{}, together with the fast reduction mechanism
included in \Coq{}'s kernel~\cite{Gregoire:2002:CIS:581478.581501}.

These two ingredients are also used behind the scene by tactics
implementing verified decision procedures. For instance, we make
extensive use of proof commands dedicated to the normalization of
algebraic expressions like the \lstinline|field| tactic for rational
fractions, and the \lstinline|ring| tactic for
polynomials~\cite{mahboubi:hal-00819484}. The field tactic generates
proof obligations describing sufficient conditions for the
simplifications it made. In our case, these conditions in turn are
solved using the \C{lia} decision procedure for linear
arithmetics~\cite{Besson:2006:FRA:1789277.1789281}.

These tactics work by first converting formulas in the goal into
instances of appropriate data structures, suitable for larger scale
computation.  This pre-processing, hidden to the user, is
performed by extra-logical code that is part of the internal
implementation of these tactics. The situation is different when a
computational step in a proof requires the evaluation of a formula at
a given argument, and when both the formula and the argument are
described using proof-oriented, inefficient representations. In that
case, for instance for evaluating terms in a given sequence, we used
the \CoqEAL{} library~\cite{cohen:hal-01113453}, which provides an
infrastructure automating the conversion between different
data-structures and algorithms used to model the same mathematical
objects, like different representations of integers or different
implementations of a matrix product. Note that although the \CoqEAL{}
library itself depends on a library for big numbers, which provides
direct access in \Coq{} to \Ocaml{}'s library for arbitrary-precision,
arbitrary-size signed integers, the present proof does not need this
feature.

\section{Outline of the proof}
\label{sec:proof}
There exists several other proofs of Apéry's theorem. Notably,
Beukers~\cite{Beukers-1979-NIZ} published an elegant proof, based
on integrals of pseudo-Lengendre polynomials, shortly after Apéry's
announcement.
According to Fischler's survey~\cite{Fischler04}, all these proofs share a
common structure. They rely on the asymptotic behavior of the sequence
$\ell_n$, the least common multiple of integers between $1$ and $n$,
and they proceed by exhibiting two sequences of rational numbers $a_n$
and $b_n$, which have the following properties:
\begin{enumerate}
\item For a sufficiently large $n$:
$$a_n \in \bZ \quad \textrm{ and } 2\ell_n^3b_n \in \bZ;$$
\item The sequence $\delta_n = a_n\zeta(3) - b_n$ is such that:
$$\limsup_{n\rightarrow \infty}|2\delta_n|^ {\frac{1}{n}} \leq (\sqrt{2}
-1)^ 4;$$
\item\label{it:lim0} For an infinite number of values $n$, $\delta_n \neq 0$.
\end{enumerate}
Altogether, these properties entail the irrationality of
$\zeta(3)$. Indeed, if we
suppose that there exists $p,q\in\bZ$ such that
$\zeta(3) = \frac{p}{q}$, then $2q\ell_n^ 3\delta_n$ is an integer when $n$ is
large enough. One variant of the Prime Number theorem states
that $\ell_n =e^{n\,(1+o(1))}$ and since $(\sqrt{2} -1)^4e^3 < 1$, the
sequence $2q\ell_n^3\delta_n$ has a zero limit, which contradicts the third
property. Actually, the Prime Number theorem can be replaced by
a weaker estimation of the asymptotic behavior of $\ell_n$, that can
be obtained by more elementary means.

\begin{lem}\label{lem:hanson}
Let $\ell_n$ be the least common multiple of integers $1, \dots, n$,
then
$$\ell_n = O(3^n) .$$
\end{lem}
Since we still have $(\sqrt{2} -1)^43^3 < 1$, this
observation~\cite{Hanson-1972-PP,Feng-2005-SEP} is enough to
conclude. Section~\ref{sec:hanson} discusses the formal proof of
Lemma~\ref{lem:hanson}, an ingredient which was missing at the time of
writing the previous report on this work~\cite{chyzak:hal-00984057}.

In our formal proof, we consider the pair of sequences proposed by
Apéry in his proof~\cite{Apery-1979-IDD,vanderPoorten-1979-PEM}:
\begin{equation}\label{eq:a-and-b}
a_n = \sum_{k=0}^n  {\tbinom nk}^2 {\tbinom{n+k}k}^2, \qquad
b_n = a_n z_n +
\sum_{k=1}^n\sum_{m=1}^k \frac{(-1)^{m+1}{\tbinom nk}^2 {\tbinom{n+k}k}^2}
{2 m^3 \binom nm \binom{n+m}m}
\end{equation}
where $z_n$ denotes $\sum_{m=1}^n \frac1{m^3}$, as already used in
Proposition \ref{prop:z3def}.

By definition, $a_n$ is a positive integer for any $n\in\bN$. The
integrality of $2\ell_n^3b_n$ is not as straightforward, but rather
easy to see as well: each summand in the double sum defining $b_n$ has
a denominator that divides~$2\ell_n^3$. Indeed, after a suitable
re-organization in the expression of the summand, using standard
properties of binomial coefficients, this follows easily from the
following slightly less standard property:
\begin{lem}\label{lem:binom-div-lcmn}
For any integers $i,j,n$ such that $1 \leq j \leq i \leq n$,
$j\binom{i}{j} \text{ divides } \ell_n$.
\end{lem}
\begin{proof}
For $i,j,n$ such that $1 \leq j \leq i \leq n$, the proof goes by
showing that for any prime $p$, the $p$-adic valuation of
$j\binom{i}{j}$ is at most that of $\ell_n$. Let us fix a prime number
$p$. Let $t_p(i)$ be the largest integer $e$ such that $p^e \leq
i$. By definition, and since $i \leq n$, we thus have $p^{t_p(i)}|\ell_n$
and  so $t_p(i) \leq v_p(\ell_n)$.  Hence it suffices to prove that
$v_p(\binom{i}{j}) \leq t_p(i) - v_p(j)$.
Using Lemma~\ref{lem:valfact}, and because $j \leq i < p ^ {t_p(i) +1}$, we have:

$$ v_p(\binom{i}{j}) = \sum_{k = 1}^{t_{p}(i)}\floor{\frac{i}{p^k}} -
(\sum_{k = 1}^{t_{p}(i)}\floor{\frac{j}{p^k}} + \sum_{k =
  1}^{t_{p}(i)}\floor{\frac{(i - j)}{p ^ k}} )$$ Remember that for
$a,b\in\bN$, $\floor{\frac{a}{b}}$ is just $a$ modulo $b$. Now for $1
\leq k \leq v_p(j)$, and because $p^k | j$, we have
$\floor{\frac{i}{p^k}} = \floor{\frac{j}{p^k}} + \floor{\frac{(i -
    j)}{p ^ k}}$, and thus:

    \begin{eqnarray*}
      v_p(\binom{i}{j})  & = &
      \sum_{k = v_p(j) + 1}^{t_{p}(i)}\floor{\frac{i}{p^k}} -
      (\sum_{k = v_p(j) + 1}^{t_{p}(i)}\floor{\frac{j}{p^k}} +
      \sum_{k = v_p(j) + 1}^{t_{p}(i)}\floor{\frac{(i - j)}{p ^ k}} )\\
      & = &
      \sum_{k = 1}^{t_{p}(i) - v_p(j)}\floor{\frac{i}{p^{ v_p(j) + k}}} -
      (\sum_{k = 1}^{t_{p}(i) - v_p(j)}\floor{\frac{j}{p^{ v_p(j) + k}}} +
      \sum_{k = 1}^{t_{p}(i) - v_p(j)}\floor{\frac{(i - j)}{p ^ {v_p(j) + k}}} )
      \\
    \end{eqnarray*}
Now for any $1 \leq k \leq t_{p}(i) - v_p(j)$, we have:
    $$\floor{\frac{i}{p^{ v_p(j) + k}}} \leq \floor{\frac{j}{p^{
      v_p(j) + k}}} + \floor{\frac{(i - j)}{p ^ {v_p(j) + k}}} +1$$
Summing both sides for $k$ from $1$ to $t_{p}(i) - v_p(j)$ and using
the previous identity for $v_p(\binom{i}{j})$
eventually proves that $v_p(\binom{i}{j}) \leq t_p(i) - v_p(j)$, which
concludes the proof.
\end{proof}
The rest of the proof is a study of the sequence
$\delta_n = a_n\zeta(3) - b_n$. It not difficult to see that
$\delta_n$~tends to zero, from the formulas defining the sequences $a$
and $b$, but we also need to prove that it does so
fast enough to compensate for~$\ell_n^3$, while being positive.
In his original proof, Apéry derived the latter facts by
combining the definitions of the sequences $a$ and $b$ with the study
of a mysterious recurrence relation. Indeed, he
made the surprising claim that Lemma~\ref{lem:rec-apery} holds:
\begin{lem}\label{lem:rec-apery}
For $n\geq 0$, the sequences $(a_n)_{n\in\bN}$ and $(b_n)_{n\in\bN}$ satisfy the
\emph{same\/} second-order recurrence:
\begin{equation}\label{eq:apery-rec}
  (n+2)^3 y_{n+2} - (17n^2+51n+39)(2n+3) y_{n+1} + (n + 1)^3 y_n = 0 .
\end{equation}
\end{lem}
Equation~\ref{eq:apery-rec} is a typical example of a linear
recurrence equation with polynomial coefficients and standard
techniques~\cite{Salvy-2003-AAV,vanderPoorten-1979-PEM} can be used to
study the asymptotic behavior of its solutions. Using this recurrence
and the initial conditions satisfied by $a$ and $b$, one can thus
obtain the two last properties of our criterion, and conclude with the
irrationality of $\zeta(3)$. For the purpose of our formal proof, we
devised an elementary version of this asymptotic study, mostly based
on variations on the presentation of van der
Poorten~\cite{vanderPoorten-1979-PEM}. We detail this part of the
proof in Section~\ref{sec:asymptotics}.

Using only Equation~\ref{eq:apery-rec}, even with sufficiently many initial
conditions, it would not be easy to obtain the first property of our
criterion, about the integrality of $a_n$ and $b_n$ for a large enough
$n$. In fact, it would also be difficult to prove that the sequence $\delta$
tends to zero: we would only know that it has a finite limit, and how
fast the convergence is. By contrast, it is fairly easy to obtain
these facts from the explicit closed forms given in Formula~\ref{eq:a-and-b}.

The proof of Lemma~\ref{lem:rec-apery} was by far the most difficult
part in Apéry's original
exposition. In his report~\cite{vanderPoorten-1979-PEM}, van der
Poorten describes how he, with other colleagues, devoted
significant efforts to this verification after having attended the
talk in which Apéry exposed his result for the first time.
Actually, the proof of Lemma~\ref{lem:rec-apery} boils down to
a routine calculation using the two auxiliary sequences
$U_{n,k}$ and $V_{n,k}$, themselves defined in terms of
$\lambda_{n,k}=\binom{n}{k}^2\binom{n+k}{k}^2$ (with $\lambda_{n,k}=0$
if $k<0$ or $k>n$):
\begin{eqnarray*}
U_{n,k}  & =  & 4(2n +1)(k(2k+1) -(2n +1)^2)\lambda_{n,k},\\
V_{n,k}  & = & U_{n,k} \left(\sum_{m=1}^{n}\frac{1}{m^3} +
             \sum_{m=1}^k\frac{(-1)^{m-1}}{2m^3\binom{n}{m}\binom{n+m}{m}}\right)
  +\frac{5(2n+1)k(-1)^{k-1}}{n(n+1)}\binom{n}{k}\binom{n+k}{k}
\end{eqnarray*}
The key idea is to compute telescoping sums for $U$ and $V$. For
instance, we have:
\begin{equation}\label{eq:telA}
U_{n,k} - U_{n,k-1} = (n+1)^3 \lambda_{n+1,k} - (34n^3 + 51n^2 + 27n
+ 5)\lambda_{n,k} + n^3 \lambda_{n-1,k}
\end{equation}
Summing Equation~\ref{eq:telA} on $k$ shows that the sequence $a$
satisfies the recurrence relation of Lemma~\ref{lem:rec-apery}. A
similar calculation proves the analogue for $b$, using telescoping
sums of the sequence $V$.

Not only is the statement of Formula~\ref{eq:apery-rec} difficult to
discover: even when this recurrence is given, finding the suitable
auxiliary sequences $U$ and $V$ by hand is a difficult task. Moreover,
there is no other known way of proving Lemma~\ref{lem:rec-apery} than
by exhibiting this nature of certificates. Fortunately, the sequences
$a$ and $b$ belong in fact to a class of objects well known in the
fields of combinatorics and of computer-algebra.  Following seminal work of
Zeilberger's~\cite{Zeilberger-1990-HSA}, algorithms have been designed
and implemented in computer-algebra systems, which are able to obtain
linear recurrences for these sequences. For instance the Maple
package \Mgfun{} (distributed as part of the
\Algolib{}~\cite{Algolib} library) implements these algorithms, among
others. Basing on this implementation, Salvy wrote a Maple
worksheet~\cite{Salvy-2003-AAV} that follows Apéry's original method
but interlaces Maple calculations with human-written parts. In
particular, this worksheet illustrates how parts of this proof,
including the discovery of Apéry's mysterious recurrence, can be
performed by symbolic computations. Our formal proof of
Lemma~\ref{lem:rec-apery} follows an approach similar to the one of
Salvy. It is based on calculations performed using the
\Algolib{} library, and certified a posteriori. This part
of the formal proof is discussed in Section~\ref{sec:d-finite}.

\section{Algorithms, Recurrences and Formal Proofs}
\label{sec:algos}

\textit{This section quotes and summarizes an earlier
publication~\cite{chyzak:hal-00984057},
describing a joint work with Chyzak and Tassi.}

\vspace{2ex}

Lemma~\ref{lem:rec-apery} is the bottleneck in Apéry's proof.  Both
sums $a_n$ and~$b_n$ in there are instances of \emph{parameterized
  summation}: they follow the pattern $F_n =
\sum_{k=\alpha(n)}^{\beta(n)} f_{n,k}$ in which the summand~$f_{n,k}$,
potentially the bounds, and thus the sum, depend on a parameter~$n$.
This makes it appealing to resort to the algorithmic paradigm of
\emph{creative telescoping}, which was developed for this situation in
computer algebra.

\subsection{Recurrences as a data structure for sequences}
\label{sec:d-finite}

A fruitful idea from the realm of computer algebra is to represent
sequences not explicitly, such as the univariate $(n!)_n$ or the
bivariate $(\binom{n}{k})_{n,k}$, but by a system of linear
recurrences of which they are solutions such as $\{ u_{n+1} = (n+1)
u_n \}$ or $\{ u_{n+1,k} = \frac{n+1}{n+1-k} u_{n,k} , u_{n,k+1} =
\frac{n-k}{k+1} u_{n,k} \}$, accompanied with sufficient initial
conditions. Sequences which can be represented in such a way are
called \emph{$\partial$-finite}. The finiteness property of their
definition makes algorithmic most operations under which the class of
$\partial$-finite sequences is stable.

In the specific bivariate case which interests us, let $S_n$ be the
shift operator in $n$ mapping a sequence $(u_{n,k})_{n,k}$ to
$(u_{n+1,k})_{n,k}$ and similarly, let $S_k$ map $(u_{n,k})_{n,k}$ to
$(u_{n,k+1})_{n,k}$.
Linear recurrences canceling a sequence $f$ can be seen as elements of
a non-commutative ring of polynomials with coefficients in
$\mathbb{Q}(n,k)$, and with the two indeterminates $S_n$ and $S_k$,
with the action $(P\cdot f)_{n,k} = \sum_{(i,j) \in I} p_{i,j}(n,k)
f_{n+i,k+j}$, where subscripts denote evaluation.  For example
for~$f_{n,k} = \binom nk$, the previous recurrences once rewritten as
equalities to zero can be represented as~$P\cdot f = 0$ for $P = S_n -
\frac{n+1}{n+1-k}$ and $P = S_k - \frac{n-k}{k+1}$, respectively.

Computer algebra gives us algorithms to produce canceling operators
for operations such as the addition or product of two
$\partial$-finite sequences, using for both its inputs and output a
Gröbner basis as a canonical way to represent the set of their
canceling operators, which gives some uniqueness guarantees.

The case of summing a sequence~$(f_{n,k})$
into a parameterized sum $F_n = \sum_{k=0}^n f_{n,k}$ is more
involved: it follows the method of \emph{creative
        telescoping\/}~\cite{Zeilberger-1991-MCT}, in two stages.  First, an
\emph{algorithmic\/} step determines pairs~$(P,Q)$ satisfying
\begin{equation}\label{eq:P-eq-Delta_k-Q}
P\cdot f = (S_k - 1)Q \cdot f
\end{equation}
with $P \in \bQ(n)[S_n]$ and~$Q \in \cA$.  To continue with our example
$f_{n,k} = \binom nk$, we could choose $P = S_n - 2$ and~$Q = S_n -
1$.  Second, a \emph{systematic\/} but not fully algorithmic step follows:
summing~\eqref{eq:P-eq-Delta_k-Q} for~$k$ between~0 and~$n +
\deg P$ yields
\begin{equation}\label{eq:optimistic-ct}
(P\cdot F)_n = (Q\cdot f)_{k = n + \deg P + 1} - (Q\cdot f)_{k = 0} .
\end{equation}
Continuing with our binomial example,
summing~\eqref{eq:P-eq-Delta_k-Q} for $k$ from~0 to~$n+1$ (and taking
special values into account) yields $\sum_{k=0}^{n+1} \binom{n+1}k - 2
\sum_{k=0}^n \binom nk = 0$, a special form
of~\eqref{eq:optimistic-ct} with right-hand side canceling to
zero. This tells us that the sequence $(\sum_{k=0}^n \binom nk)_{n \in
  \bN}$ is a solution of the same recurrence $P = S_n - 2$ as
$(2^n)_{n \in \bN}$: a simple check of initial values gives us the
identity $\forall n \in \mathbb{N}, 2^n = \sum_{k=0}^n \binom nk$.
The formula~\eqref{eq:optimistic-ct} in fact assumes several
hypotheses that hold not so often in practice; this will be
formalized by Equation~\eqref{eq:guarded-pre-ct} below.

\subsection{Apéry's sequences are $\partial$-finite constructions}

The sequences $a$ and~$b$ in~\eqref{eq:a-and-b} are $\partial$-finite:
they have been announced to be solutions of~\eqref{eq:apery-rec}.  But
more precisely, they can be viewed as constructed from “atomic”
sequences by operations under which the class of~$\partial$-finite
sequences is stable.  This is summarized in Table~\ref{eq:seq-defs}.

\begin{table}[ht]
        \begin{center}
          {\renewcommand{\arraystretch}{1.5}
                  \begin{tabular}{cc@{\quad}ccc}
                          \toprule
                                {\bfseries ~step~} & {\bfseries explicit form} &
                                {\bfseries ~GB~} & {\bfseries operation} &
                                {\bfseries input(s)} \\\midrule
                                1 & $c_{n,k} = {\binom nk}^2 {\binom{n+k}k}^2$ & $C$ & direct & \\
                                2 & $a_n = \sum_{k=1}^n c_{n,k}$ & $A$ & creative telescoping & $C$ \\
                                3 & $d_{n,m} = \frac{(-1)^{m+1}}{2 m^3 \binom nm \binom{n+m}m}$ & $D$ & direct & \\
                                4 & $s_{n,k} = \sum_{m=1}^k d_{n,m}$ & $S$ & creative telescoping & $D$ \\
                                5 & $z_n = \sum_{m=1}^n \frac1{m^3}$ & $Z$ & direct & \\
                                6 & $u_{n,k} = z_n + s_{n,k}$ & $U$ & addition & $Z$ and $S$ \\
                                7 & $v_{n,k} = c_{n,k} u_{n,k}$ & $V$ & product & $C$ and $U$ \\
                                8 & $b_n = \sum_{k=1}^n v_{n,k}$ & $B$ & creative telescoping & $V$ \\
                                \bottomrule
                        \end{tabular}}
                        \medskip
                        \caption{Construction of $a_n$ and $b_n$\label{eq:seq-defs}: At each
                                step, the Gröbner basis named in column~GB, which annihilates the
                                sequence given in explicit form, is obtained by the corresponding
                                operation \emph{on ideals}, with input(s) given on the last column.}
                \end{center}
        \end{table}
In this table, Gröbner bases are systems of recurrence operators: at
each line in the table, the sequence given in explicit form is a
solution of the system of recurrences described by the operators in
the Gröbner basis column.
Note that in fact none of these results rely on the \emph{specific}
sequences in the explicit form: at each step, a new Gröbner basis is
obtained from known ones, the ones that are cited in the input column.
The table can also be read bottom-up for the purpose of verification:
the Gröbner basis obtained at a given step can be verified using
\emph{only} the
Gröbner bases obtained at some previous steps, all the way down to $C$
and $D$. These operators describe a more general class of (germs of)
sequences than just the explicit sequences used in this table, thus
initial conditions are needed to describe a precise sequence.

\subsection{Provisos and sound creative telescoping}
\label{sec:provisos}

We illustrate the process of verifying candidate new recurrences using
known ones on the example of Pascal's triangle rule.
One can “almost prove” Pascal's triangle rule
using only the following recurrences, satisfied by binomial coefficients:
$$u_{n+1,k} = \frac{n+1}{n+1-k}u_{n,k} \quad \textrm{ and } \quad
  u_{n,k+1} = \frac{n-k}{k+1}u_{n,k}.$$
Indeed, we have:
\begin{equation*}
\binom{n+1}{k+1} - \binom n{k+1} - \binom nk =
\left(\frac{n+1}{n-k}\frac{n-k}{k+1} - \frac{n-k}{k+1} - 1\right)
\binom nk = 0 \times \binom nk = 0 ,
\end{equation*}
\emph{but this requires $k \neq -1$ and\/~$k \neq n$}. Therefore, this
does not prove Pascal's rule for all $n$ and~$k$. The phenomenon is
general: computer algebra is unable to take denominators into
account. This incomplete modeling of sequences by algebraic objects
may cast doubt on these computer-algebra proofs, in particular when it
comes to the output of creative-telescoping algorithms.

By contrast, in our formal proofs, we augmented the recurrences with
provisos that restrict their applicability. In this setting, we
validate a candidate identity like the Pascal triangle rule by a
normalization modulo the elements of a Gröbner basis plus a
verification that this normalization only involves legal instances of
the recurrences.  In the case of creative telescoping,
Eq.~\eqref{eq:P-eq-Delta_k-Q} takes the form:
\begin{equation}\label{eq:guarded-pre-ct}
(n,k)\notin \Delta \Rightarrow (P\cdot f_{\_,k})_n = (Q\cdot
  f)_{n,k+1} - (Q\cdot f)_{n,k} ,
\end{equation}
where $\Delta \subset \bZ^2$ guards the relation and where
$f_{\_,j}$~denotes the univariate sequence obtained by specializing
the second argument of~$f$ to~$j$. Thus our formal analogue of
Eq.~\eqref{eq:optimistic-ct} takes this restriction into account and
has the shape
\begin{equation}\label{eq:ct-crux}
\begin{split}
(P\cdot F)_n & = \Bigl( (Q\cdot f)_{n,n+\beta+1} - (Q\cdot f)_{n,\alpha} \Bigr)
+ \sum_{i=1}^r\sum_{j=1}^i p_i(n) \, f_{n + i,n + \beta + j} \\
& + \sum_{\alpha \leq k \leq n + \beta\ \wedge\ (n,k) \in \Delta}
(P\cdot f_{\_,k})_n - (Q\cdot f)_{n,k+1} + (Q\cdot f)_{n,k} ,
\end{split}
\end{equation}
for $F$ the sequence with general term $F_n =
\sum_{k=\alpha}^{n+\beta}f_{n,k}$ and where $P = \sum\limits_{i=0}^r
p_i(n) S_n^i$.

The last term of the right-hand side, which we will call the singular
part, witnesses the possible partial domain of validity of
relation~\eqref{eq:guarded-pre-ct}.  Thus the operator~$P$ is a valid
recurrence for the sequence~$F$ if the right-hand side of
Eq.~\eqref{eq:ct-crux} normalizes to zero, at least outside of an
algebraic locus that will guard the recurrence.

\subsection{Generated Operators, hand-written provisos, and formal proofs}


For each step in Table~\ref{eq:seq-defs}, we make use of the data
computed by the Maple session in a systematic way, using
pretty-printing code to express this data in \Coq{}.  As mentioned in
Section~\ref{sec:provisos}, we manually annotate each operator
produced by the computer-algebra program with provisos and turn it
this way into a conditional recurrence predicate on sequences.  In our
formal proof, each step in Table~\ref{eq:seq-defs} consists in proving
that some conditional recurrences on a composed sequence can be proved
from some conditional recurrences known for the arguments of the
operation.

These steps are far from automatic, mainly because the singular part
yields terms which have to be reduced manually through trial-and-error
using a Gröbner basis of annihilators for $f$, but also because we
have to show that some rational-function coefficients of the remaining
instances of $f$ are zero. This is done through a combination of the
\lstinline|field| and \lstinline|lia| proof commands, helped by some
factoring of denominators pre-obtained in Maple.


\subsection{Composing closures and reducing the order of $B$}

In order to complete the formal proof of Lemma~\ref{lem:rec-apery}, we
verify formally that each sequence involved in the construction of
$a_n$ and $b_n$ is a solution of the corresponding Gröbner system of
annotated recurrence, starting from $c_n$, $d_n$, and $z_n$ all the
way to the the final conclusions. This proves that $a_n$ is a solution
of the recurrence~\eqref{eq:apery-rec} but only provides a recurrence
of order four for $b_n$. We then prove that $b$~also satisfies the
recurrence~\eqref{eq:apery-rec} using four evaluations
$b_0,b_1,b_2,b_3$.

\section{Consequences of Apéry's recurrence}
\label{sec:asymptotics}
In this section, we detail the elementary proofs of the properties
obtained as corollaries of Lemma~\ref{lem:rec-apery}. We recall, from
Section~\ref{sec:proof}, that these properties describe the asymptotic
behavior of the sequence $\delta_n = a_n\zeta(3) -b_n$, with:
\begin{equation}
a_n = \sum_{k=0}^n  {\tbinom nk}^2 {\tbinom{n+k}k}^2, \qquad
b_n = a_n z_n +
\sum_{k=1}^n\sum_{m=1}^k \frac{(-1)^{m+1}{\tbinom nk}^2 {\tbinom{n+k}k}^2}
{2 m^3 \binom nm \binom{n+m}m}
\end{equation}
Throughout the section, we use the vocabulary and notations of Cauchy sequences
numbers, as introduced in Section~\ref{ssec:realcomplex}. For instance, we
have:
\begin{lem}\label{lem:eventbovera}
  For any $\varepsilon$, eventually $|z_n - \frac{b_n}{a_n}| < \varepsilon$.
\end{lem}
\begin{proof} Easy from the definition of $z$, $a$ and
  $b$. \end{proof}

\begin{cor}\label{cor:z3bovera}
The sequence $(\frac{b_n}{a_n})_{n\in\bN}$ is a Cauchy sequence, which
is Cauchy equivalent to $(z_n)_{n\in\bN}$.
\end{cor}

\noindent The formal statement corresponding to
Lemma~\ref{lem:eventbovera} is:
\begin{lstlisting}
Lemma |*z3seq_b_over_a_asympt*| : {asympt e : n / |z3seq n - b_over_a_seq n| < e}.
\end{lstlisting}
where \C{b_over_a_seq n} represents $\frac{b_n}{a_n}$. The notation
\C{\{asympt e : i / P\}}, used in this formal statement, comes from
the external library for Cauchy
sequences~\cite{cohen:hal-00671809}. In the expression \C{\{asympt e :
  i / P\}}, \C{asympt} is a keyword, and both \C{e} and \C{i} are
names for variables bound in the term \C{P}. This expression unfolds
to the term \C{(asympt1 (fun e i => P))}, a dependent pair that ensures the
existence of an explicit witness that property \C{P} asymptotically
holds:
\begin{lstlisting}
Definition |*asympt1*| R (P : R -> nat -> Prop) :=
   {m : R -> nat | forall (eps : R) (i : nat), 0 < eps -> m eps <= i -> P eps i}.
\end{lstlisting}
The formalization of Corollary~\ref{cor:z3bovera} then comes in three
steps: first the proof that $(\frac{b_n}{a_n})_{n\in\bN}$ is a Cauchy
sequence, as formalized by the \C{creal_aiom} predicate:
\begin{lstlisting}
Corollary |*creal_b_over_a_seq*| : creal_axiom b_over_a_seq.
\end{lstlisting}
This formal proof is a one-liner, because the corresponding general
argument, a sequence that is asymptotically close to a Cauchy
sequence will itself satisfy the Cauchy property, is already present in
the library.  Then, the latter proof is used to forge an inhabitant of
the type of rational Cauchy sequences, which just amounts to pairing
the sequence \C{b_over_a_seq : nat -> rat} with the latter proof:
\begin{lstlisting}
Definition |*b_over_a*| : {creal rat} := CReal creal_b_over_a_seq.
\end{lstlisting}
Now we can state the proof of equivalence between the two Cauchy
sequences, i.e., between the two corresponding terms of type
\C{\{creal rat\}}:

\begin{lstlisting}
Fact |*z3_eq_b_over_a*| : (z3 == b_over_a)%CR.
\end{lstlisting}

The proof of the latter fact is again a one-liner, with no additional
mathematical content added to lemma \C{creal_b_over_a_seq}, but it
provides access to automation based on setoid rewriting facilities

The \MC{} libraries do not cover any topic of analysis, and even the most
basic definitions of transcendental functions like the exponential or
the logarithm are not available. However, it is possible to obtain
the required properties of the sequence $\delta$ by very elementary
means, and almost all these elementary proofs can be inferred from a
careful reading and a combination of Salvy's
proof~\cite{Salvy-2003-AAV} and of van der Poorten's
description~\cite{vanderPoorten-1979-PEM}.

Following van der Poorten, we introduce an auxiliary sequence
$(w_n)\in\bQ^n$, defined as:
$$w_n =
\begin{vmatrix}
b_{n+1} & a_{n+1}\\
b_n & a_n
\end{vmatrix}
= b_{n+1}a_n - a_{n+1}b_n.
$$
The sequence $w$ is called a \emph{Casoratian}: as $a$ and $b$ are
solutions of a same linear recurrence relation~\eqref{eq:apery-rec} of
order 2, this can be seen as a discrete analogue of the Wronskian for
linear differential systems. For example, $w$ satisfies a recurrence
relation of order 1, which provides a closed form for $w$:

\begin{lem}\label{lem:wform}
For $n \geq 2$, $w_n = \frac{6}{(n + 1)^3}$.
\end{lem}
\begin{proof}
Since $a$ and $b$ satisfy the recurrence relation~\ref{eq:a-and-b},
$w$ satisfies the relation:
$$\forall k \in \bN, (k + 2)^3w_{k+1} - (k + 1)^3w_k = 0.$$
The result follows from the computation of $w_0$.

\end{proof}
From this formula, we can obtain the positivity of the sequence
$\delta$, and an evaluation of its asymptotic behavior in terms of the
sequence $a$.
\begin{cor}\label{cor:posp}
For any $n\in\bN$ such that $2 \leq n$, $0 < \zeta(3) - \frac{b_n}{a_n}$.
\end{cor}
\noindent The formal statement corresponding to Corollary~\ref{cor:posp} is the
following:
\begin{lstlisting}
Lemma |*lt0_z3_minus_b_over_a*| (n : nat) : 2 <= n -> (0%:CR < z3 - (b_over_a_seq n)%:CR)%CR.
\end{lstlisting}
Note the postfix \C{CR} tag which enforces that inside the corresponding
parentheses, notations are interpreted in the scope associated with
Cauchy sequences: in particular, the order relation (on Cauchy
sequences) is the one described in Section~\ref{ssec:realcomplex}.

Term \C{(b_over_a_seq n : rat)} is the rational number
$\frac{a_n}{b_n}$, which is casted as a Cauchy real, the corresponding
constant sequence, using the postfix \C{\%:CR}, so as to be subtracted to the
Cauchy sequence \C{z3}. This proof in particular benefits
from setoid rewriting using equivalences like
\C{z3_eq_b_over_a}, the formal counterpart of Corollary~\ref{cor:z3bovera}.

\begin{proof}
Since $\zeta(3)$ is Cauchy equivalent to $(\frac{b_n}{a_n})_{n\in\bN}$, it is sufficient to show that for any
$k < l$, we have $0 < \frac{b_l}{a_l} - \frac{b_k}{a_k}$. Thus it is
sufficient to observe that for any $k$, we have $0 <
\frac{b_{k+1}}{a_{k+1}} - \frac{b_k}{a_k}$, which follows from
Lemma~\ref{lem:wform}.

\end{proof}
\begin{cor}\label{cor:delta}
$\zeta(3) - \frac{b_n}{a_n} = \bigO(\frac{1}{a_n^2})$.
\end{cor}
\begin{proof}
Since $\zeta(3) = \frac{b}{a}$, it is sufficient to show that there
exists a constant $K$, such that for any $k < l$, $\frac{b_l}{a_l} -
\frac{b_k}{a_k} \leq \frac{K}{a_k^2}$. But since $a$ is an increasing
sequence, Lemma~\ref{lem:wform} proves that for any $k < l$,
$\frac{b_l}{a_l} - \frac{b_k}{a_k} \leq \sum_{i=k}^{l-1}
\frac{w_i}{a_i a_{i+1}} \leq \sum_{i=k}^{l-1}\frac{6}{(i+1)^3a_k^2}
\leq \frac{K}{a_k^2}$, for any $K$ greater than $6\cdot\zeta(3)$.
\end{proof}
The last remaining step of the proof is to show that the sequence $a$
grows fast enough. The elementary version of Lemma~\ref{lem:asympta}
is based on a suggestion by F.~Chyzak.
\begin{lem}\label{lem:asympta}
$33^n = \bigO(a_n)$.
\end{lem}
\begin{proof}
Consider the auxiliary sequence $\rho_n = \frac{a_{n+1}}{a_n}$. Since
$\rho_{51}$ is greater than $33$, we only need to show that the
sequence $\rho$ is increasing. For the sake of readability, we
denote $\mu_n$ and $\nu_n$ the fractions coefficients of the
recurrence satisfied by $a$, obtained from
Equation~\ref{eq:apery-rec} after division by its leading coefficient.
Thus $a$ satisfies the recurrence relation:
$$a_{n+2} - \mu_n a_{n+1} + \nu_n = 0.$$
For $n\in\bN$, we also introduce the function $h_n(x) = \mu_n +
\frac{\nu_n}{x}$, so that $\rho_{n+1} = h_n(\rho_n)$.
The polynomial $P_n(x) = x^2 - \mu_nx + \nu_n$ has two distinct
roots $x_n' <x_n$, and the formula describing the roots of polynomials
of degree 2 show that $0 < x_n' < 1 < x_n$ and that the sequence $x_n$
is increasing. But since $h_n(x) - x = -\frac{P_n(x)}{x}$, for
$1 < x < x_n$, we have $h_n(x) > x$. A direct
recurrence shows that for any $n \geq 2$, $\rho_n \in [1, x_n]$, which
concludes the proof. 
\end{proof}
In the formal proof of Lemma~\ref{lem:asympta}, the computation of
$\rho_{51}$ was made possible by using the \CoqEAL{} library, as 
already mentioned in Section~\ref{subsec:computations}. This proof
also requires a few symbolic
computations that are a bit tedious to perform by hand: in these
cases, we used Maple as an oracle to massage algebraic expressions,
before formally proving the correctness of the simplification. This
was especially useful to study the roots $x_n'$ and $x_n$ of $P_n$.

We can now conclude with the limit of the sequence $\ell_n^3\delta_n$,
under the assumption that $\ell_n = \bigO(3^n)$.
\begin{cor}
\label{cor:ldelta_0}
 $\lim\limits_{n \to \infty}(\ell_n^3\delta_n) = 0$.
\end{cor}
\begin{proof}
Immediate, since $\delta_n =
\bigO(\frac{1}{a_n})$ by Corollary~\ref{cor:delta}, and  $\ell_n^3 =
\bigO((3^3)^n)$, and $3^3 < 33$.
\end{proof}
In the next Section, we describe the proof of the last remaining
assumption, about the asymptotic behavior of $\ell_n$.
\section{Asymptotics of $lcm(1, ... , n)$}
\label{sec:hanson}

For any integer $1 \leq n$, let $\ell_n$ denote the least common
multiple $lcm(1, ... , n)$ of the integers no greater than $n$. By
convention, we set $\ell_0 = 1$. The asymptotic behavior of the
sequence $(\ell_n)$ is a classical corollary of the Prime Number
Theorem. A sufficient estimation for the present proof can actually
but obtained as a direct consequence, using an elementary remark about
the $p$-adic valuations of~$\ell_n$.

\begin{rem}\label{rmk:pvall}
 For any prime
number $p$, the integer $p^ {v_p(\ell_n)}$ is the highest power of $p$
not exceeding $n$, so that:
$$v_p(\ell_n) = \floor{\log_p(n)}.$$
\end{rem}
\begin{proof}
Noticing that $v_p(lcm(a,b)) = \max(v_p(a),v_p(b))$, we see by
induction on $n$ that $v_p(\ell_n) = \max\limits_{i=1}^n
v_p(i)$. Recall from Section~\ref{subsec:nat} that $\lfloor \log_p(n)
\rfloor$ is a notation for the greatest integer $\alpha$ such that $p
^ \alpha \leq n$. Since $\alpha = v_p(p ^ \alpha)$, we have $\alpha
\leq v_p(\ell_n)$. Now suppose that $v_p(\ell_n) = v_p(i)$ for some $i
\in \{1, \dotsc , n\}$. Then $i = p^{v_p(i)} q$ with $gcd(p,q) = 1$ so
that $p^{v_p(\ell_n)} = p^{v_p(i)} \leq i \leq n$ and thus
$v_p(\ell_n) \leq \alpha$. This proves that $v_p(\ell_n) = \alpha$.
\end{proof}

%
By Remark~\ref{rmk:pvall}, $\ell_n$ can hence be written as $\prod_{p\leq
  n} p ^ {\floor{\log_p(n)}}$ and therefore:
$$\ln (\ell_n) = \sum_{p\leq n}\ln(p^{\floor{\log_p(n)}}) \leq \sum_{p\leq n}\ln(n).$$
If $\pi(n)$ is
the number of prime numbers no greater than $n$, we hence have:
$$\ln(\ell_n) \leq \pi(n) \ln(n).$$ The Prime Number theorem states that
$\pi(n)\sim\frac{n}{\ln(n)}$; we can thus conclude that:
$$\ell_n = \bigO( e^{n}).$$
Note that this estimation is in fact rather precise, as in fact:
$$\ell_n \sim  e^{n\,(1+o(1))}.$$

J.~Avigad and his co-authors provided the
first machine-checked proof of the Prime Number
theorem~\cite{Avigad:2007:FVP:1297658.1297660}, which was considered
at the time as a formalization \emph{tour de force}. Their
formalization is based on a proof attributed to A.~Selberg and
P.~Erdös. Although the standard proofs of this theorem use tools from
complex analysis like contour integrals, their choice was guided by
the corpus of formalized mathematics available for the Isabelle proof
assistant, or the limits thereof. Although less direct, the proof by
A.~Selberg and P.~Erdös is indeed more elementary and avoids complex
analysis completely.

\subsection{Statement, Notations and Outline}
\label{subsec:hanson_statement}

In order to prove Corollary~\ref{cor:ldelta_0} in
Section~\ref{sec:asymptotics}, we only need to resort to
Lemma~\ref{lem:hanson}, i.e., to the fact that:
         $$\ell_n = \bigO(3 ^ n).$$
This part of the proof was left as an assumption in our previous
report~\cite{chyzak:hal-00984057}. This weaker description of the
asymptotic behavior of $(\ell_n)$ was
in
fact known before the first proofs of the Prime Number theorem but our
formal proof is a variation on an elementary proof proposed by Hanson
\cite{Hanson-1972-PP}. 

The idea of the proof is to replace the study of $\ell_n$ by that of
another sequence $C(n)$. The latter is defined as a multinomial
coefficient depending on elements of a fast-growing sequence
$\alpha$. The fact that $\prod_{i=1}^n \alpha_i^{1 / \alpha_i} < 3$
independently of $n$ then allows to show that $C(n) = \bigO(3^n)$.

\subsection{Proof}
\label{subsec:hanson_proof}

Define the sequence $(\alpha_n)_{n \in \mathbb{N}}$ by $\alpha_0 = 2$, and
$\alpha_{n+1} = \alpha_1 \alpha_2 \dotsm \alpha_n + 1$ for $n \ge 1$.
By an induction on $n$, this is equivalent to $\alpha_{n+1} =\alpha_n^2 -
\alpha_n + 1$.  For $n,k \in \mathbb{N}$, let
\[C(n,k) = \frac{n!}
   {\lfloor n / \alpha_1 \rfloor ! \lfloor  n / \alpha_2 \rfloor !
     \dotsm \lfloor n / \alpha_k\rfloor!}.\]
As soon as $\alpha_k \geq n$, $C(n,k)$ is independent of $k$ and we denote
$C(n) = C(n,k)$ for all such $k$. Hanson directly defines $C(n)$ as a
limit, but we found this to be inconvenient to manipulate in the
proof. Moreover, most inequalities stated on $C(n)$ actually hold for
$C(n,k)$ with little or no more hypotheses. The following technical
lemma will be useful in the study of this sequence.

\begin{lem}
\label{suminv}
For $k \in \mathbb{N}$,
\[\sum_{i=1}^{k} \frac{1}{\alpha_i} = \frac{\alpha_{k+1} - 2}{\alpha_{k+1} - 1} < 1
\textrm{ and thus for } x \in \mathbb{Q} \textrm{ with } x \geq 1,
\lfloor x \rfloor > \sum_{i=1}^{k} \left \lfloor \frac{x}{\alpha_i}
\right\rfloor. \]
\end{lem}
\begin{proof}
The proof is done by induction and relies on the fact that if $a \in
\mathbb{Q}$ and $m \in \mathbb{N}^{+}$, we have $ \left\lfloor
\frac{a}{m} \right\rfloor = \left\lfloor \frac{\lfloor a \rfloor}{m}
\right\rfloor.$
\end{proof}
Notice that by Lemma~\ref{suminv},
$\sum\limits_{i=0}^{k} \left\lfloor \frac{n}{\alpha_i}\right\rfloor < n$ and thus:
$$C(n,k) = \binom{\sum\limits_{i=0}^{k} \left\lfloor \frac{n}{\alpha_i} \right\rfloor}{\lfloor n / \alpha_1 \rfloor , \lfloor n /
  \alpha_2 \rfloor, \dotsc, \lfloor n / \alpha_k\rfloor} \frac{n!}{\left(
\sum\limits_{i=0}^{k} \left\lfloor \frac{n}{\alpha_i} \right\rfloor
\right)!}.$$

\noindent In particular, $C(n,k) \in \mathbb{N}$.
The goal is now to show that $\ell_n \leq C(n) < K \cdot 3^n$ for some $K$.

In the following, for $n,k, p\in \mathbb{N}$ and $p$ prime, we denote
$\beta_p(n,k)$ for $v_p(C(n,k))$.
\begin{lem}
For all $n,k \in \bN$, with $1 \leq n$ and $p$ prime, $\beta_p(n,k) \geq \lfloor
\log_p(n) \rfloor = v_p(\ell_n)$. Therefore $C(n,k) \ge \ell_n$ for all $n$.
\end{lem}

\begin{proof}
The proof uses Lemma~\ref{lem:valfact}.
\begin{align*}
\beta_p(n,k) & = v_p(n!) - \sum_{i=1}^{k} v_p(\lfloor
n / \alpha_i \rfloor !) \\
& = \sum_{i=1}^{\lfloor \log_p(n) \rfloor } \lfloor n / p^i \rfloor -
\sum_{i=1}^{k} \sum_{j=1}^{\lfloor \log_p(\lfloor \frac{n}{\alpha_i}
  \rfloor) \rfloor} \left\lfloor \frac{ n } {\alpha_i p^j} \right\rfloor
\\
& = \sum_{i=1}^{\lfloor \log_p(n) \rfloor } \left(\lfloor n / p^i
\rfloor - \sum_{j=1}^{k} \left\lfloor \frac{ \lfloor \frac{n}{p^i}
  \rfloor} {\alpha_j} \right\rfloor \right) \\
& \ge \sum_{i=1}^{\lfloor
  \log_p(n) \rfloor} 1 \textrm{ (because $\sum
  \frac{1}{\alpha_i} < 1$ by Lemma~\ref{suminv})}.
\end{align*}
Since $\ell_n = \prod\limits_{p \leq n}^{} p ^{\lfloor \log_p(n)
  \rfloor}$ from Remark~\ref{rmk:pvall}, we get $\ell_n \leq C(n,k) =
\prod\limits_{p \leq n}^{} p ^{\beta_p(n,k)}$.
\end{proof}

\begin{lem}
\label{analysis_ineq}
For $i \ge 1$ and $n \geq \alpha_i$,
\[\frac{\left (\frac{n}{\alpha_i} \right)^{\frac{n}{\alpha_i} }}{ \left \lfloor \frac{n}{\alpha_i} \right\rfloor ^{\lfloor \frac{n}{\alpha_i} \rfloor }}
< \left( \frac{10 \, n}{\alpha_i}\right)^{\frac{\alpha_i - 1}{\alpha_i}}.
\]
\end{lem}
\begin{proof}
If $n = \alpha_i$, we have $1 < \sqrt{10} \leq 10^{\frac{\alpha_i - 1}{\alpha_i}}$,
hence the result.  Otherwise $n > \alpha_i$: let us write
$n = b \alpha_i + r$, with $0 \leq r < \alpha_i$ the Euclidean
division of $n$ by $\alpha_i$. We have:
$$\frac{n - \alpha_i + 1}{\alpha_i} = b - 1 + \frac{r + 1}{\alpha_i}
\leq b.$$
Since $b = \left \lfloor \frac{n}{\alpha_i} \right\rfloor$, it follows
that $\frac{n - \alpha_i + 1}{\alpha_i} \leq \left \lfloor
\frac{n}{\alpha_i} \right\rfloor$. Now for $\frac{1}{2}\leq x$, the function $x^x$ is increasing, thus since
$\frac{1}{2} \leq \frac{1}{\alpha_i} \leq \frac{n - \alpha_i + 1}{\alpha_i}$,
we  deduce that $ \left(\frac{n - \alpha_i +
  1}{\alpha_i}\right)^{\frac{n - \alpha_i + 1}{\alpha_i}} \leq \left \lfloor
\frac{n}{\alpha_i} \right\rfloor ^{\lfloor \frac{n}{\alpha_i} \rfloor
}$, and we hence have:
\begin{align*}
\frac{\left (\frac{n}{\alpha_i} \right)^{\frac{n}{\alpha_i} }}{ \left \lfloor
  \frac{n}{\alpha_i} \right\rfloor ^{\lfloor \frac{n}{\alpha_i}
    \rfloor }} & \leq
\frac{\left (\frac{n}{\alpha_i} \right)^{\frac{n}{\alpha_i} }}{\left(\frac{n -
    \alpha_i + 1}{\alpha_i}\right)^{\frac{n - \alpha_i + 1}{\alpha_i}}} = {\left(1 +
  \frac{\alpha_i - 1}{{n - \alpha_i + 1}}\right)^{\frac{n - \alpha_i + 1}{\alpha_i - 1}}}
^ {\frac{\alpha_i - 1}{\alpha_i}} \left(\frac{n}{\alpha_i}\right)^{\frac{\alpha_i -
    1}{\alpha_i}}.
\end{align*}
The first operand in the last expression is of the shape
$\left({\left({{1 + \frac{1}{x}}}\right)^x}\right)^{\frac{\alpha_i -
    1}{\alpha_i}}$, where $x$ is a positive rational. We showed using only
elementary properties of rational numbers, like the binomial formula
or the summation formula for geometric progressions, that for any such
positive rational number $x$, $\left(1 + \frac{1}{x} \right)^x < 10 $, hence
the result. Note that we only needed that there exist a constant $K >
0$ such that $\left(1 + \frac{1}{x} \right)^x < K$.
\end{proof}

Of course, using elementary real analysis allows for the tighter bound
$e$, which was used in Hanson's paper, but this bound is irrelevant
for the final result. We can now establish the following bound on $C(n,k)$:

\begin{lem}
\label{C(n)_ineq_1}
For $k \ge 1$ and $n \ge 2$,
\[ C(n,k) < \frac{n^n}{{\lfloor \frac{n}{\alpha_1}\rfloor}^{\lfloor \frac{n}{\alpha_1}\rfloor}
  \dotsm  {\lfloor \frac{n}{\alpha_k}\rfloor}^{\lfloor \frac{n}{\alpha_k}\rfloor}}.
\]
\end{lem}
\begin{proof}

First observe that if $m = m_1 + \dots + m_k$ where $m$ and the $m_i$
are (not all zero) non-negative integers, we have because of the
multinomial theorem:

\begin{align}
\label{multi_ineq}
\left( m_1 + \dots + m_k\right)^m \ge \binom{m}{m_1, \dots, m_k} m_1^{m_1} \dotsm m_k^{m_k}.
\end{align}

Let $k \ge 1$, and define: 
\[ t = \sum_{i=1}^{k} \left\lfloor \frac{n}{\alpha_i} \right\rfloor. \]
Then $t < n$ by Lemma~\ref{suminv}. We have:
\begin{align}
\label{c_n_fun_of_t}
C(n,k) & = n \cdot (n-1) \dotsm (t+1) \binom{t}{\lfloor n / \alpha_1 \rfloor , \lfloor  n / \alpha_2 \rfloor , \dots ,\lfloor n / \alpha_k\rfloor}.
\end{align}
Because of equation (\ref{multi_ineq}), we know that

\begin{align}
\label{ineq_multi_t}
\binom{t}{\lfloor n / \alpha_1 \rfloor , \lfloor n / \alpha_2 \rfloor , \dots
  ,\lfloor n / \alpha_k\rfloor} & \leq \frac{t^t}{\lfloor n / \alpha_1 \rfloor
  ^{\lfloor n / \alpha_1 \rfloor} \lfloor n / \alpha_2 \rfloor ^{\lfloor n / \alpha_2
    \rfloor} \dots \lfloor n / \alpha_k\rfloor^{\lfloor n / \alpha_k\rfloor}}.
\end{align}

From equations (\ref{c_n_fun_of_t}) and (\ref{ineq_multi_t}) we deduce that
\begin{align*}
C(n,k) & < \frac{n^{n}}{\lfloor n / \alpha_1 \rfloor ^{\lfloor n / \alpha_1
    \rfloor} \lfloor n / \alpha_2 \rfloor ^{\lfloor n / \alpha_2 \rfloor} \dots
  \lfloor n / \alpha_k\rfloor^{\lfloor n / \alpha_k\rfloor}}.
  \qedhere
\end{align*}
\end{proof}

\begin{lem}\label{lem_k_loglogn}
Let $k \ge 3$, $n\in \mathbb{N}$. If $\alpha_k \leq n$ then
\[k < \lfloor \log_2 {\lfloor \log_2 n \rfloor} \rfloor + 2.\]
\end{lem}
\begin{proof}
First observe by a simple induction that for all $k \geq 3$, $\alpha_{k} >
2^{2^{k-2}} + 1$ so that $k-2 < \lfloor \log_2 \left( \lfloor \log_2
\alpha_k \rfloor \right) \rfloor \leq \lfloor \log_2 \left( \lfloor \log_2
n \rfloor \right) \rfloor$.  
\end{proof}

\begin{lem}\label{ineq2_5}
Let $k \ge 1$, $n\in \mathbb{N}$. If $\alpha_k \leq n$,
\[ C(n,k) < \frac{n^n (\frac{10 \, n}{\alpha_1})^{\frac{\alpha_1 - 1}{\alpha_1}} (\frac{10 \, n}{\alpha_2})^{\frac{\alpha_2 - 1}{\alpha_2}} \dots (\frac{10 \, n}{\alpha_k})^{\frac{\alpha_k - 1}{\alpha_k}}} {\left(\frac{10 \, n}{\alpha_1}\right)^\frac{10 \, n}{\alpha_1} \left(\frac{10 \, n}{\alpha_1}\right)^\frac{10 \, n}{\alpha_1} \dots \left(\frac{10 \, n}{\alpha_k}\right)^\frac{10 \, n}{\alpha_k}}.\]
\end{lem}
\begin{proof}
The result is straightforward by combining Lemmas~\ref{analysis_ineq}
and~\ref{C(n)_ineq_1}.  
\end{proof}

\begin{lem}
\label{lem:wk}
Let $w_k = \prod\limits_{i=1}^{k} \alpha_i^{\frac{1}{\alpha_i}}$, $k \ge
1$. Then $w_k$ is increasing and there exists $w \in \mathbb{R}$, with
\[ w < 2.98, \]
such that $w_k < w$.
\end{lem}
\begin{proof}
The sequence $w_k$ is increasing because ${\alpha_i}^{\frac{1}{\alpha_i}} > 1$
(because $\alpha_i > 1$). Since $\alpha_i^2 > \alpha_{i+1} > (\alpha_i - 1)^2$, one can
see that for $i \geq 3$, $\alpha_{i+1}^{\frac{1}{\alpha_{i+1}}} <
\sqrt{\alpha_i^{\frac{1}{\alpha_i}}}$, so that for all $k \geq 1$ and $l \geq
0$,

\[ w_{k+l} \leq \prod\limits_{i=1}^{k} \alpha_i^{\frac{1}{\alpha_i}} \cdot \alpha_{k+1}^{\frac{1}{\alpha_{k+1}}\sum_{i=0}^{l} \frac{1}{2^i}} \leq w_k \cdot {\alpha_{k+1}^{\frac{2}{\alpha_{k+1}}}}. \]

We establish by an elementary external computation verified in \Coq{}
that $\alpha_1^{\frac{1}{\alpha_1}} < \frac{283}{200}$, $\alpha_2^{\frac{1}{\alpha_2}} <
\frac{1443}{1000}$, $\alpha_3^{\frac{1}{\alpha_3}} < \frac{1321}{1000}$,
$\alpha_4^{\frac{1}{\alpha_4}} < \frac{273}{250}$ and $\alpha_5^{\frac{1}{\alpha_5}} <
\frac{201}{200}$. From the bound above with $k = 4$ we get $w < w_4
\cdot \alpha_5^{\frac{2}{\alpha_5}} \leq \frac{5949909309448377}{2 \cdot
  10^{15}} < 2.98$.  
\end{proof}

\begin{rem}\label{obs1}
For $k \geq 1$, we have
\begin{align*}
\frac{\alpha_1 - 1}{\alpha_1} + \frac{\alpha_2 - 1}{\alpha_2} + \dots + \frac{\alpha_k -
  1}{\alpha_k} = k - 1 + \frac{1}{\alpha_{k+1} - 1}.
\end{align*}
\end{rem}

\begin{proof}
It is a direct consequence of Lemma~\ref{suminv}.
\end{proof}
Note that the statement of Remark~\ref{obs1} actually corrects a typo
in the original paper.

\begin{thm}\label{thm_ineq}
If $\alpha_k \leq n < \alpha_{k+1}$,
\[C(n,k) = C(n) < (10 \, n)^{k - \frac{1}{2}}  w^{n+1}. \]
\end{thm}
\begin{proof}
From Lemma~\ref{ineq2_5}, recall that we have
\begin{align*}
C(n,k) & <  \frac{n^n (\frac{10 \, n}{\alpha_1})^{\frac{\alpha_1 - 1}{\alpha_1}}
  (\frac{10 \, n}{\alpha_2})^{\frac{\alpha_2 - 1}{\alpha_2}} \dots (\frac{10 \,
    n}{\alpha_k})^{\frac{\alpha_k - 1}{\alpha_k}}}
{\left(\frac{n}{\alpha_1}\right)^\frac{n}{\alpha_1}
  \left(\frac{n}{\alpha_2}\right)^\frac{n}{\alpha_2} \dots
  \left(\frac{n}{\alpha_k}\right)^\frac{n}{\alpha_k}}\\
& =  \frac{n^n (10 \,
  n)^{\left(\sum\limits_{i=1}^{k} \frac{\alpha_i - 1}{\alpha_i}\right)} \left(
  \prod_{i=1}^{k} \alpha_i^{\frac{1}{\alpha_i} } \right)^n } {n^{n
    \sum_{i=1}^{k} \frac{1}{\alpha_i}} \prod_{i=1}^{k} \alpha_i^{\frac{\alpha_i -
      1}{\alpha_i}}}.
\end{align*}
It can be seen using Lemma~\ref{suminv} that:
\[ n ^{n \left( 1 - \sum_{i=1}^{k} \frac{1}{\alpha_i} \right) } \leq n. \]
Thus
\begin{align*}
C(n,k) & < n \frac{(10 \, n)^{k - 1 + \frac{1}{\alpha_{k+1} - 1}} w_k^n } {
  \prod_{i=1}^{k} \alpha_i^{\frac{\alpha_i - 1}{\alpha_i}}} \textrm{
  (by Remark~\ref{obs1})} \\
& \leq n \frac{(10 \, n)^{k - 1 + \frac{1}{\alpha_{k+1} -
      1}} w^n } { \prod_{i=1}^{k} \alpha_i^{\frac{\alpha_i - 1}{\alpha_i}}} \textrm{
  because $w_k \leq w$.}
\end{align*}

Since $n < \alpha_{k+1} = 1 + \prod_{i=1}^{k} \alpha_i$, $n \leq \prod_{i=1}^{k}
\alpha_i$, and we have
\[ \prod_{i=1}^{k} \alpha_i^{\frac{\alpha_i - 1}{\alpha_i}} = \frac{\prod\limits_{i=1}^{k} \alpha_i}{w_k} \geq \frac{n}{w_k}. \]

Thus
\begin{align*}
C(n,k) & < {(10 \, n)^{k - 1 + \frac{1}{\alpha_{k+1} - 1}} w^n w_k }
\\
& \leq { (10 \, n)^{k - \frac{1}{2}} w^{n+1}} \textrm{ as $\alpha_{k+1}
  \geq 3$ and $w_k \leq w$}.\qedhere
\end{align*}
\end{proof}

We can now prove Lemma~\ref{lem:hanson}:
\begin{proof}
  We have:
\[ \ell_n / 3^n \leq C(n,k) / 3^n = (10 \, n)^{k - \frac{1}{2}} \left(\frac{w}{3}\right)^{n+1}.\]
Remembering that $k < \lfloor \log_2 {\lfloor \log_2 n \rfloor}
\rfloor + 2$ and $w < 3$, it is elementary to show that the quantity
on the right is eventually decreasing to $0$ and therefore bounded,
which proves the result.  We once again make use of the fact that
$\left(1 + \frac{1}{x} \right) ^ x$ is bounded in the course of this
elementary proof. 
\end{proof}

\section{Conclusion} 
\label{sec:conclusions_formal_proof}

We are not aware of a comprehensive, reference, formal proof library
on the topic of number theory, in any guise. The most comprehensive
work in this direction is probably the \IsaHOL{} library on analytic
number theory contributed by Eberl~\cite{DBLP:conf/itp/Eberl19}, which
covers a substantial part of an introductory textbook by
Apostol~\cite{apostol}. This library is based on an extensive corpus
in complex analysis initially formalized by Harrison in the
\HOLLight{} prover, and later ported to the \IsaHOL{} prover by
Paulson and Li.  Formal proofs also exists of a few emblematic
results. The elementary fact that $\sqrt{2}$ is irrational was used as
an example problem in a comparative study of the styles of various
theorem provers~\cite{Wiedijk:2006:SPW:1121735}, including \Coq{}. The
Prime Number theorem was proved formally for the first time by Avigad
\emph{et al}.~\cite{Avigad:2007:FVP:1297658.1297660}, using the \IsaHOL{}
prover and later by Harrison~\cite{harrison-pnt} with the
\HOLLight{} prover. Shortly after the submission of the first version
of the present paper, Eberl verified~\cite{Zeta_3_Irrational-AFP} Beuker's
proof of Apéry's theorem~\cite{Beukers-1979-NIZ}, using the \IsaHOL{}
prover, and relying on the Prime Number
theorem to derive the asymptotic properties of $\ell_n$.  Bingham was
the first to formalize a proof that $e$ is
transcendental~\cite{JFR2269}, with the \HOLLight{} prover. Later,
Bernard \emph{et al.} formally proved the transcendence of both $\pi$ and $e$
in \Coq{}~\cite{DBLP:journals/corr/BernardBRS15}.

Some of the
ingredients needed in the present proof are however not specific to
number theory. For instance, we here use a very basic, but sufficient, infrastructure
to represent asymptotic behaviors. But ``big Oh'', also called
Bachman-Landau, notations have been discussed more extensively by
Avigad \emph{et al.}~\cite{Avigad:2007:FVP:1297658.1297660} in the context of
their formal proof of the Prime Number Theorem, and by Boldo \emph{et
al.} \cite{DBLP:journals/corr/abs-1005-0824}, for the asymptotic
behavior of real-valued continuous functions.
Affeldt \emph{et al.} designed a sophisticated infrastructure for equational
reasoning in \Coq{} with Bachman-Landau
notation~\cite{DBLP:journals/jfrea/AffeldtCR18}, which relies on a
non-constructive choice operator.
Another example of such a secondary topic is the theory of multinomial
coefficients, which is also relevant to combinatorics, and which is
also defined by Hivert in his \Coq{} library
Coq-Combi~\cite{Coq-Combi}. However, up to our knowledge this library
does not feature a proof of the generalized Newton identity.

Harrison~\cite{harrison-wz} presented a way to produce
rigorous proofs
from certificates produced by Wilf-Zeilberger certificates, by
seeing sequences as limits of complex functions. His method applies to
the sequence $a$, and can in principle prove that it satisfies the recurrence
equation~(\ref{eq:apery-rec}). However, this method does not allow for
a proof that $b$ satisfies the recurrence
relation~(\ref{eq:apery-rec}), because the summand
is itself a sum but not a hypergeometric one. Up to our knowledge, there
is no known way today to justify the output of the efficient algorithms of
creative telescoping used here without handling a trace of provisos.

The idea to use computer algebra software (CAS) as an oracle
outputting a certificate to be checked by a theorem prover, dubbed a
\emph{skeptic's} approach, was first introduced by Harrison and
Théry~\cite{HarrisonThery-1998-SAC}. It is based on the observation
that CAS are very efficient albeit not always correct, while theorem
provers are sound but slow. This technique thus takes the best of both
worlds to produce reliable proofs requiring large scale
computations. In the case of \Coq{}, this viewpoint is especially
fruitful since the kernel of the proof assistant includes efficient
evaluation mechanisms for the functional programs written inside the
logic~\cite{Gregoire:2002:CIS:581478.581501}. Notable successes based
on this idea include the use of Pocklington certificates to check
primality inside \Coq{}~\cite{Gregoire2006} or external computations
of commutative Gröbner bases, with applications for instance in
geometry~\cite{DBLP:journals/corr/abs-1007-3615}. Delahaye and Mayero
proposed~\cite{Delahaye:2005:DAE:1740741.1740843} to use CAS to help
experimenting with algebraic expressions inside a proof assistant,
before deciding what to prove and how to prove it. Unfortunately,
their tool was not usable in our case, where algebraic expressions are
made with operations that come from a hierarchy of structures.

Organizing the cooperation of a CAS and a proof assistant sheds light
on their respective differences and drawbacks.  The initial motivation
of this work was to study the algorithms used for the automatic
discovery and proof of recurrences.  Our hope was to be able to craft
an automated tool providing formal proofs of recurrences, by using the
output of these algorithms, in a skeptical way. This plan did not work
and Section~\ref{sec:algos} illustrates the impact of confusing the
rational fractions manipulated by symbolic computations with their
evaluations, which should be guarded by conditions on the
denominators. On the other hand, proof assistants are not yet equipped
to manipulate the large expressions imported during the cooperation,
even those which are of a small to moderate size for the standards of
computer algebra systems. For instance, we have highlighted in our
previous report~\cite{chyzak:hal-00984057} the necessity to combine
two distinct natures of data-structures in our libraries: one devoted
to formal proofs, which may use computation inside the logic to ease
bureaucratic steps, and one devoted to larger scale computations,
which provides a fine-grained control on the complexity of
operations. The later was crucial for the computations involved in the
normalization of ring expressions during the a posteriori verification
of computer-algebra produced recurrences. But it was also instrumental
in our proof of Lemma~\ref{lem:asympta}.

Incidentally, our initial formal proof of Lemma~\ref{lem:wk} also
involved this nature of calculations, with rather larger
numbers\footnote{For the current standards of proof assistants.}. Indeed, the
proof requires bounding the five first values of sequence $\alpha_n^{1
  / a_n}$ and the straightforward strategy involves intermediate
computations with integers with about $4160$ decimal digits. Following
a suggestion by one of the anonymous referees, we now use a less naive
formula to bound $\alpha_5$. This dramatically reduces the size of the
numbers involved, to the price of some additional manual bureaucracy
in the proof script, mostly for evaluating binomial coefficients
without an appropriate support at the level of libraries.

The \Coq{} proof assistant is not
equipped with a code generation feature akin to what is offered, for
instance, in the \IsaHOL{}
prover~\cite{DBLP:conf/esop/HupelN18}. In principle, it is possible to
plug in \IsaHOL{} formal developments the result of computations executed by
external, generated programs that are verified down to
machine code. In the present formal proofs, computation are instead
carried \emph{inside} the logic, using the Calculus of Inductive
constructions as a programming language. Such an approach is possible
in \Coq{} because its proof-checker includes an efficient mechanism
for evaluating these functional
programs~\cite{Gregoire:2002:CIS:581478.581501}. Automating the
correctness proofs of the program transformations required for the
sake of efficiency is on-going research. In the present work, we have
used the  \CoqEAL{} library~\cite{cohen:hal-01113453}, which is itself based on a 
plugin for parametricity proofs~\cite{DBLP:conf/csl/KellerL12}. The
variant proposed by Tabareau \emph{et
al.}~\cite{DBLP:journals/pacmpl/TabareauTS18} might eventually help
improving the extensibility of the refinement framework, which is a
key issue.

The interfaces of proof assistants are also notoriously
less advanced than those of modern computer algebra. For instance,
reasoning by transitivity on long chains of equalities/inequalities is
often cumbersome in \Coq{}, because of the limited support for
selecting terms in an expression and for reasoning by transitivity.
The Lean theorem prover~\cite{deMoura2015} features a \verb|calc|
environment for proofs on transitive relations which might be used as
a first step in this direction.

On several occasions in this work, we wrote more elementary versions
of the proofs than what we had found in the texts we were
formalizing. We partly agree with
Avigad~\emph{et~al.}~\cite{Avigad:2007:FVP:1297658.1297660} when they write that this
can be both frustrating and enjoyable: on one hand, it can illustrate
the lack of mathematical libraries for theorems which mathematicians
would find simple, such as elementary analysis for studying the
asymptotics of sequences as in Section~\ref{sec:asymptotics}. Ten
years later, ``the need for elementary workarounds'' is still present,
despite his fear that it would ``gradually fade, and with it, alas,
one good reason for investing time in such
exercises''\cite{Avigad:2007:FVP:1297658.1297660}.  On the other hand,
this need gives an opportunity to better understand the minimal scope
of mathematical theories used in a proof, with the help of a
computer. For instance, it was not clear to us at first that we could
manage to completely avoid the need to define transcendental
functions, or to avoid defining the constant $e$, base of the natural
logarithm, to formalize Hanson's paper~\cite{Hanson-1972-PP}.
This minimality however comes at the price of arguably more pedestrian
calculations.

\section*{Acknowledgment}
We wish to thank the anonymous reviewers, Alin Bostan, Frédéric
Chyzak, Georges Gonthier, Marc Mezzarobba, and Bruno Salvy for their
comments and suggestions.  We also thank Cyril Cohen, Pierre Roux and
Enrico Tassi for their help, in particular with the libraries this
work depends on.

\bibliographystyle{alpha}      
\bibliography{biblio}   

\newcommand{\etalchar}[1]{$^{#1}$}
\begin{thebibliography}{dMKA{\etalchar{+}}15}

\bibitem[ACR18]{DBLP:journals/jfrea/AffeldtCR18}
Reynald Affeldt, Cyril Cohen, and Damien Rouhling.
\newblock Formalization techniques for asymptotic reasoning in classical
  analysis.
\newblock {\em J. Formalized Reasoning}, 11(1):43--76, 2018.

\bibitem[ADGR07]{Avigad:2007:FVP:1297658.1297660}
Jeremy Avigad, Kevin Donnelly, David Gray, and Paul Raff.
\newblock A formally verified proof of the prime number theorem.
\newblock {\em ACM Trans. Comput. Logic}, 9(1), December 2007.

\bibitem[alg13]{Algolib}
Algolib.
\newblock \url{http://algo.inria.fr/libraries/}, 2013.
\newblock Version 17.0. For Maple~17.

\bibitem[Ap{\'e}79]{Apery-1979-IDD}
Roger Ap{\'e}ry.
\newblock Irrationalité de $\zeta(2)$ et $\zeta(3)$.
\newblock {\em Astérisque}, 61, 1979.
\newblock Société Mathématique de France.

\bibitem[Apo76]{apostol}
Tom~M. Apostol.
\newblock {\em Introduction to Analytic Number Theory}.
\newblock Undergraduate Texts in Mathematics. Springer, 1976.

\bibitem[BBRS15]{DBLP:journals/corr/BernardBRS15}
Sophie Bernard, Yves Bertot, Laurence Rideau, and Pierre{-}Yves Strub.
\newblock Formal proofs of transcendence for e and {\(\pi\)} as an application
  of multivariate and symmetric polynomials.
\newblock {\em CoRR}, abs/1512.02791, 2015.

\bibitem[BCF{\etalchar{+}}10]{DBLP:journals/corr/abs-1005-0824}
Sylvie Boldo, Fran{\c{c}}ois Cl{\'{e}}ment, Jean{-}Christophe Filli{\^{a}}tre,
  Micaela Mayero, Guillaume Melquiond, and Pierre Weis.
\newblock Formal proof of a wave equation resolution scheme: the method error.
\newblock {\em CoRR}, abs/1005.0824, 2010.

\bibitem[BCP03]{DBLP:journals/jfp/BartheCP03}
Gilles Barthe, Venanzio Capretta, and Olivier Pons.
\newblock Setoids in type theory.
\newblock {\em J. Funct. Program.}, 13(2):261--293, 2003.

\bibitem[Bes07]{Besson:2006:FRA:1789277.1789281}
Fr{\'e}d{\'e}ric Besson.
\newblock Fast reflexive arithmetic tactics the linear case and beyond.
\newblock In {\em {TYPES'06}}, LNCS, pages 48--62, Berlin, Heidelberg, 2007.
  Springer-Verlag.

\bibitem[Beu79]{Beukers-1979-NIZ}
F.~Beukers.
\newblock A note on the irrationality of {$\zeta (2)$} and {$\zeta (3)$}.
\newblock {\em Bull. London Math. Soc.}, 11(3):268--272, 1979.

\bibitem[Bin11]{JFR2269}
Jesse Bingham.
\newblock Formalizing a proof that e is transcendental.
\newblock {\em Journal of Formalized Reasoning}, 4(1):71--84, 2011.

\bibitem[BR01]{MR1859021}
Keith Ball and Tanguy Rivoal.
\newblock Irrationalit\'{e} d'une infinit\'{e} de valeurs de la fonction
  z\^{e}ta aux entiers impairs.
\newblock {\em Invent. Math.}, 146(1):193--207, 2001.

\bibitem[CDM13]{cohen:hal-01113453}
Cyril Cohen, Maxime D{\'e}n{\`e}s, and Anders M{\"o}rtberg.
\newblock {Refinements for Free!}
\newblock In {\em {Certified Programs and Proofs}}, pages 147 -- 162,
  Melbourne, Australia, December 2013.

\bibitem[CMSPT14]{chyzak:hal-00984057}
Fr{\'e}d{\'e}ric Chyzak, Assia Mahboubi, Thomas Sibut-Pinote, and Enrico Tassi.
\newblock {A Computer-Algebra-Based Formal Proof of the Irrationality of
  $\zeta$(3)}.
\newblock In {\em {ITP - 5th International Conference on Interactive Theorem
  Proving}}, Vienna, Austria, 2014.

\bibitem[Coh12]{cohen:hal-00671809}
Cyril Cohen.
\newblock {Construction of real algebraic numbers in Coq}.
\newblock In {\em {ITP}}, volume 7406 of {\em LNCS}. Springer, August 2012.

\bibitem[coq19]{coqfinitgroup}
{Mathematical Components Libraries}.
\newblock \url{http://math-comp.github.io/math-comp/}, 2019.

\bibitem[DM05]{Delahaye:2005:DAE:1740741.1740843}
David Delahaye and Micaela Mayero.
\newblock Dealing with algebraic expressions over a field in {C}oq using
  {M}aple.
\newblock {\em J. Symb. Comput.}, 39(5):569--592, May 2005.

\bibitem[dMKA{\etalchar{+}}15]{deMoura2015}
Leonardo de~Moura, Soonho Kong, Jeremy Avigad, Floris van Doorn, and Jakob von
  Raumer.
\newblock {\em The Lean Theorem Prover (System Description)}, pages 378--388.
\newblock Springer International Publishing, Cham, 2015.

\bibitem[Ebe19a]{Zeta_3_Irrational-AFP}
Manuel Eberl.
\newblock The irrationality of $\zeta(3)$.
\newblock {\em Archive of Formal Proofs}, December 2019.
\newblock \url{http://isa-afp.org/entries/Zeta_3_Irrational.html}, Formal proof
  development.

\bibitem[Ebe19b]{DBLP:conf/itp/Eberl19}
Manuel Eberl.
\newblock Nine chapters of analytic number theory in isabelle/hol.
\newblock In John Harrison, John O'Leary, and Andrew Tolmach, editors, {\em
  10th International Conference on Interactive Theorem Proving, {ITP} 2019,
  September 9-12, 2019, Portland, OR, {USA}}, volume 141 of {\em LIPIcs}, pages
  16:1--16:19. Schloss Dagstuhl - Leibniz-Zentrum f{\"{u}}r Informatik, 2019.

\bibitem[Fen05]{Feng-2005-SEP}
Bei-ye Feng.
\newblock A simple elementary proof for the inequality {$d_n<3^n$}.
\newblock {\em Acta Math. Appl. Sin. Engl. Ser.}, 21(3):455--458, 2005.

\bibitem[{Fis}04]{Fischler04}
St\'ephane {Fischler}.
\newblock {Irrationalit\'e de valeurs de z\^eta (d'apr\`es Ap\'ery, Rivoal,
  $\dots$).}
\newblock In {\em {S\'eminaire Bourbaki. Volume 2002/2003. Expos\'es
  909--923}}, pages 27--62, ex. Paris: Soci\'et\'e Math\'ematique de France,
  2004.

\bibitem[FSZ19]{fischler_sprang_zudilin_2019}
Stéphane Fischler, Johannes Sprang, and Wadim Zudilin.
\newblock Many odd zeta values are irrational.
\newblock {\em Compositio Mathematica}, 155(5):938–952, 2019.

\bibitem[GAA{\etalchar{+}}13]{gonthier:hal-00816699}
Georges Gonthier, Andrea Asperti, Jeremy Avigad, Yves Bertot, Cyril Cohen,
  Fran{\c c}ois Garillot, St{\'e}phane Le~Roux, Assia Mahboubi, Russell
  O'Connor, Sidi Ould~Biha, Ioana Pasca, Laurence Rideau, Alexey Solovyev,
  Enrico Tassi, and Laurent Th{\'e}ry.
\newblock {A Machine-Checked Proof of the Odd Order Theorem}.
\newblock In Sandrine Blazy, Christine Paulin, and David Pichardie, editors,
  {\em {ITP 2013, 4th Conference on Interactive Theorem Proving}}, volume 7998
  of {\em LNCS}, pages 163--179, Rennes, France, July 2013. {Springer}.

\bibitem[GGMR09]{garillot:inria-00368403}
Fran{\c c}ois Garillot, Georges Gonthier, Assia Mahboubi, and Laurence Rideau.
\newblock {Packaging Mathematical Structures}.
\newblock working paper or preprint, March 2009.

\bibitem[GL02]{Gregoire:2002:CIS:581478.581501}
Benjamin Gr{\'e}goire and Xavier Leroy.
\newblock A compiled implementation of strong reduction.
\newblock In {\em Proceedings of the Seventh ACM SIGPLAN International
  Conference on Functional Programming}, ICFP '02, pages 235--246, New York,
  NY, USA, 2002. ACM.

\bibitem[GM05]{mahboubi:hal-00819484}
Benjamin Gr{\'e}goire and Assia Mahboubi.
\newblock Proving equalities in a commutative ring done right in {C}oq.
\newblock In {\em TPHOLs 2005}, volume 3603 of {\em LNCS}, pages 98--113,
  Oxford, United Kingdom, August 2005. Springer.

\bibitem[GTW06]{Gregoire2006}
Benjamin Gr{\'e}goire, Laurent Th{\'e}ry, and Benjamin Werner.
\newblock {\em A Computational Approach to Pocklington Certificates in Type
  Theory}, pages 97--113.
\newblock Springer Berlin Heidelberg, Berlin, Heidelberg, 2006.

\bibitem[Han72]{Hanson-1972-PP}
Denis Hanson.
\newblock On the product of the primes.
\newblock {\em Canad. Math. Bull.}, 15:33--37, 1972.

\bibitem[Har09]{harrison-pnt}
John Harrison.
\newblock Formalizing an analytic proof of the {P}rime {N}umber {T}heorem.
\newblock {\em Journal of Automated Reasoning}, 43:243--261, 2009.

\bibitem[Har15]{harrison-wz}
John Harrison.
\newblock Formal proofs of hypergeometric sums (dedicated to the memory of
  andrzej trybulec).
\newblock {\em Journal of Automated Reasoning}, 55:223--243, 2015.

\bibitem[Hed98]{Hedberg:1998}
Michael Hedberg.
\newblock A coherence theorem for {M}artin-{L}\"of's type theory.
\newblock {\em Journal of Functional Programming}, 8(4):413--436, July 1998.

\bibitem[Hiv]{Coq-Combi}
Florent Hivert.
\newblock Coq-combi.
\newblock \url{https://github.com/hivert/Coq-Combi}.

\bibitem[HN18]{DBLP:conf/esop/HupelN18}
Lars Hupel and Tobias Nipkow.
\newblock A verified compiler from isabelle/hol to cakeml.
\newblock In Amal Ahmed, editor, {\em Programming Languages and Systems - 27th
  European Symposium on Programming, {ESOP} 2018, Held as Part of the European
  Joint Conferences on Theory and Practice of Software, {ETAPS} 2018,
  Thessaloniki, Greece, April 14-20, 2018, Proceedings}, volume 10801 of {\em
  Lecture Notes in Computer Science}, pages 999--1026. Springer, 2018.

\bibitem[HT98]{HarrisonThery-1998-SAC}
John Harrison and L.~Th{\'e}ry.
\newblock A skeptic's approach to combining {HOL} and {M}aple.
\newblock {\em J. Automat. Reason.}, 21(3):279--294, 1998.

\bibitem[KL12]{DBLP:conf/csl/KellerL12}
Chantal Keller and Marc Lasson.
\newblock Parametricity in an impredicative sort.
\newblock In Patrick C{\'{e}}gielski and Arnaud Durand, editors, {\em Computer
  Science Logic (CSL'12) - 26th International Workshop/21st Annual Conference
  of the EACSL, {CSL} 2012, September 3-6, 2012, Fontainebleau, France},
  volume~16 of {\em LIPIcs}, pages 381--395. Schloss Dagstuhl - Leibniz-Zentrum
  fuer Informatik, 2012.

\bibitem[MT13]{mahboubi:hal-00816703}
Assia Mahboubi and Enrico Tassi.
\newblock {Canonical Structures for the working Coq user}.
\newblock In Sandrine Blazy, Christine Paulin, and David Pichardie, editors,
  {\em {ITP 2013, 4th Conference on Interactive Theorem Proving}}, volume 7998
  of {\em LNCS}, pages 19--34, Rennes, France, July 2013. {Springer}.

\bibitem[Pot10]{DBLP:journals/corr/abs-1007-3615}
Lo{\"{\i}}c Pottier.
\newblock {Connecting Gr{\"{o}}bner Bases Programs with Coq to do Proofs in
  Algebra, Geometry and Arithmetics}.
\newblock {\em CoRR}, abs/1007.3615, 2010.

\bibitem[Riv00]{MR1787183}
Tanguy Rivoal.
\newblock La fonction z\^{e}ta de {R}iemann prend une infinit\'{e} de valeurs
  irrationnelles aux entiers impairs.
\newblock {\em C. R. Acad. Sci. Paris S\'{e}r. I Math.}, 331(4):267--270, 2000.

\bibitem[Sal03]{Salvy-2003-AAV}
Bruno Salvy.
\newblock An {A}lgolib-aided version of {A}péry's proof of the irrationality
  of $\zeta(3)$.
\newblock \url{http://algo.inria.fr/libraries/autocomb/Apery2-html/apery.html},
  2003.

\bibitem[Soz09]{DBLP:journals/jfrea/Sozeau09}
Matthieu Sozeau.
\newblock A new look at generalized rewriting in type theory.
\newblock {\em J. Formalized Reasoning}, 2(1):41--62, 2009.

\bibitem[{The}20]{the_coq_development_team_2020_3744225}
{The Coq Development Team}.
\newblock {The Coq Proof Assistant, version 8.11.0}, January 2020.

\bibitem[TTS18]{DBLP:journals/pacmpl/TabareauTS18}
Nicolas Tabareau, {\'{E}}ric Tanter, and Matthieu Sozeau.
\newblock Equivalences for free: univalent parametricity for effective
  transport.
\newblock {\em {PACMPL}}, 2({ICFP}):92:1--92:29, 2018.

\bibitem[vdP79]{vanderPoorten-1979-PEM}
Alfred van~der Poorten.
\newblock A proof that {E}uler missed: {A}p\'ery's proof of the irrationality
  of {$\zeta (3)$}.
\newblock {\em Math. Intelligencer}, 1(4):195--203, 1979.
\newblock An informal report.

\bibitem[Wie06]{Wiedijk:2006:SPW:1121735}
Freek Wiedijk.
\newblock {\em The Seventeen Provers of the World: Foreword by Dana S. Scott
  (Lecture Notes in Computer Science / Lecture Notes in Artificial
  Intelligence)}.
\newblock Springer-Verlag New York, Inc., Secaucus, NJ, USA, 2006.

\bibitem[Zei90]{Zeilberger-1990-HSA}
Doron Zeilberger.
\newblock A holonomic systems approach to special functions identities.
\newblock {\em J. Comput. Appl. Math.}, 32(3):321--368, 1990.

\bibitem[Zei91]{Zeilberger-1991-MCT}
Doron Zeilberger.
\newblock The method of creative telescoping.
\newblock {\em J. Symbolic Comput.}, 11(3):195--204, 1991.

\bibitem[Zud01]{Zudilin-2001-ONZ}
V.~V. Zudilin.
\newblock One of the numbers {$\zeta(5)$}, {$\zeta(7)$}, {$\zeta(9)$},
  {$\zeta(11)$} is irrational.
\newblock {\em Uspekhi Mat. Nauk}, 56(4(340)):149--150, 2001.

\end{thebibliography}

\end{document}